\documentclass[11pt]{article}
\usepackage[utf8]{inputenc}
\usepackage[linesnumbered,vlined]{algorithm2e}
\usepackage{amsfonts}
\usepackage{graphicx}
\usepackage{amsthm}
\usepackage{mathrsfs}
\usepackage{amssymb}
\usepackage{hyperref}
\usepackage{amsmath}

\newcommand{\qedsymb}{\qed}
\newenvironment{proofof}[1]{\begin{trivlist}
		\item[\hspace{\labelsep}{\bf\noindent Proof of #1: }]
	}{\qedsymb\end{trivlist}}

\usepackage{cite}
\usepackage{xcolor}

\usepackage{booktabs} 

\newtheorem{theorem}{Theorem}
\newtheorem{lemma}{Lemma}
\newtheorem{definition}{Definition}
\newtheorem{claim}{Claim}
\newtheorem{corollary}{Corollary}

\usepackage{xspace}

\newcommand{\clique}{\textsc{Congested Clique}\xspace}

\newcommand\ceil[1]{\lceil#1\rceil}

\usepackage{soul}
\newcommand{\alg}{\text{SMM}\xspace}
\newcommand{\algnice}{\text{SBMM}\xspace}

\newcommand{\algASBMMSubroutineOne}{ExchangeInfo\xspace}

\newcommand{\argmin}{argmin}
\newcommand{\nice}{sparsity-balanced}
\newcommand{\SIndecies}[1]{(#1-1)(n/a)+1:#1(n/a)}
\newcommand{\TIndecies}[1]{(#1-1)(n/b)+1:#1(n/b)}
\providecommand{\ceil}[1]{\left \lceil #1 \right \rceil }

\SetKwComment{tcpy}{\#~}{} 
\SetKw{Send}{send}
\SetKw{LearnEdges}{LearnEdges}
\SetKw{LearnPaths}{LearnPaths}
\SetKw{Broadcast}{broadcast}
\RestyleAlgo{boxruled}
\LinesNumbered
\newcommand{\algrule}[1][.2pt]{\par\vskip.5\baselineskip\hrule height #1\par\vskip.5\baselineskip}
\let\oldnl\nl
\newcommand{\nonl}{\renewcommand{\nl}{\let\nl\oldnl}}

\newenvironment{theorem-repeat}[1]{\begin{trivlist}
		\item[\hspace{\labelsep}{\bf\noindent Theorem~\ref{#1} (repeated) }]\it}%
	{\end{trivlist}}

\usepackage{graphicx}

\usepackage{fullpage}
\begin{document}
	
\begin{titlepage}
	
\title{Sparse Matrix Multiplication and Triangle Listing\\ in the Congested Clique Model\thanks{A preliminary version of this paper appeared in OPODIS 2018.}}

\author{Keren Censor-Hillel\footnote{Department of Computer Science, Technion. Email: ckeren@cs.technion.ac.il.}
\and Dean Leitersdorf\footnote{Department of Computer Science, Technion. Email: dean.leitersdorf@gmail.com.}
\and Elia Turner\footnote{Department of Computer Science, Technion. Email: eliaturner11@gmail.com.}}
		
			\maketitle

\begin{abstract}
			We show how to multiply two $n \times n$ matrices $S$ and $T$ over semirings in the \textsc{Congested Clique} model, where $n$ nodes communicate in a fully connected synchronous network using $O(\log{n})$-bit messages, within $O(nz(S)^{1/3} nz(T)^{1/3}/n + 1)$ rounds of communication, where $nz(S)$ and $nz(T)$ denote the number of non-zero elements in $S$ and $T$, respectively. By leveraging the sparsity of the input matrices, our algorithm greatly reduces communication costs compared with general multiplication algorithms [Censor-Hillel et al., PODC 2015], and thus improves upon the state-of-the-art for matrices with $o(n^2)$ non-zero elements. Moreover, our algorithm exhibits the additional strength of surpassing previous solutions also in the case where only one of the two matrices is such. Particularly, this allows to efficiently raise a sparse matrix to a power greater than 2. As applications, we show how to speed up the computation on non-dense graphs of $4$-cycle counting and all-pairs-shortest-paths.
			
			Our algorithmic contribution is a new \emph{deterministic} method of restructuring the input matrices in a sparsity-aware manner, which assigns each node with element-wise multiplication tasks that are not necessarily consecutive but guarantee a balanced element distribution, providing for communication-efficient multiplication.
			
			Moreover, this new deterministic method for restructuring matrices may be used to restructure the adjacency matrix of input graphs, enabling faster deterministic solutions for graph related problems. As an example, we present a new sparsity aware, \emph{deterministic} algorithm which solves the triangle listing problem in $O(m/n^{5/3} + 1)$ rounds, a complexity that was previously obtained by a \emph{randomized} algorithm [Pandurangan et al., SPAA 2018], and that matches the known lower bound of $\tilde{\Omega}(n^{1/3})$ when $m=n^2$ of [Izumi and Le Gall, PODC 2017, Pandurangan et al., SPAA 2018].
			Naturally, our triangle listing algorithm also implies triangle counting within the same complexity of $O(m/n^{5/3} + 1)$ rounds, which is (possibly more than) a \emph{cubic} improvement over the previously known \emph{deterministic} $O(m^2/n^3)$-round algorithm [Dolev et al., DISC 2012].
\end{abstract}

\thispagestyle{empty}
\end{titlepage}

	\section{Introduction}	
	\label{section:intro}
	Matrix multiplication is a fundamental algebraic task, with abundant applications to various computations. The value of the exponent $\omega$ of matrix multiplication, that is, the value $\omega$ for which $\Theta(n^{\omega})$ is the complexity of matrix multiplication, is a central question in algebraic algorithms~\cite{Strassen69,CoppersmithW90,Williams12}, and is currently known to be bounded by $2.3728639$ \cite{LeGall14a}.
	
	The work of Censor-Hillel et al.~\cite{Censor-HillelKK15} recently showed that known matrix multiplication algorithms for the \emph{parallel} setting can be adapted to the distributed \clique model, which consists of $n$ nodes in a fully connected synchronous network, limited by a bandwidth of $O(\log{n})$ bits per message. Subsequently, this significantly improved the state-of-the-art for a variety of tasks, including triangle and $4$-cycle counting, girth computations, and (un)weighted/(un)directed all-pairs-shortest-paths (APSP). This was followed by the beautiful work of Le Gall~\cite{LeGall16}, who showed how to efficiently multiply rectangular matrices, as well as multiple independent multiplication instances. These led to even faster algorithms for some of the tasks, such as weighted or directed APSP, as well as fast algorithms for new tasks, such as computing the size of the maximum matching.
	
	In many cases, multiplication is required to be carried out for \emph{sparse} matrices, and this need has been generating much effort in designing algorithms that are faster given sparse inputs, both in sequential (e.g.,~\cite{YusterZ05,LeGall12,KaplanSV06,AmossenP09,LeGallU18}) and parallel (e.g., ~\cite{BallardBDGLST13,BulucG11,BulucG12,AzadBBDGSTW16,BallardDKS16,LazzaroVHS17,KoanantakoolABM16,SolomonikBVH17}) settings.
	
	In this paper we focus our attention on the task of multiplying sparse matrices in the \clique model, providing a novel \emph{deterministic} algorithm with a round complexity which depends on the sparsity of the input matrices.
	
		An immediate application of our algorithm is faster counting of $4$-cycles.
	Moreover, a prime feature of our algorithm is that it speeds up matrix multiplication even if \emph{only one} of the input matrices is sparse. The significance of this ability stems from the fact that the product of sparse matrices may be non-sparse, which in general may stand in the way of fast multiplication of more than two sparse matrices, such as raising a sparse matrix to a power that is larger than 2. Therefore, this property of our algorithm enables, for instance, a fast algorithm for computing APSP in the \clique model.
	We emphasize that, unlike the matrix multiplication algorithms of~\cite{Censor-HillelKK15}, we are not aware of a similar sparse matrix multiplication algorithm existing in the literature of parallel settings.
	
	Furthermore, we leverage our techniques to obtain a deterministic algorithm for sparsity-aware triangle listing in the \clique model, in which each triangle needs to be known to some node. This problem has been tackled (implicitly) in the \clique model for the first time by Dolev et al. \cite{DolevLP12}, providing two deterministic algorithms. Later,~\cite{Pandurangan0S16,IzumiG17} showed a $\tilde{\Omega}(n^{1/3})$ lower bound in general graphs. Pandurangan et al.~\cite{Pandurangan0S16} showed a \emph{randomized} triangle listing algorithm, with the same round complexity as we obtain.
	
	\subsection{Our contribution}
	For a matrix $A$, let $nz(A)$ be its number of nonzero elements.
	Our main contribution is an algorithm called \alg (Sparse Matrix Multiplication), for which we prove the following.
	
	\newcommand{\TheoremMain}
	{
		\sloppy{
			Given two $n \times n$ matrices $S$ and $T$, Algorithm \alg ~deterministically computes the product $P = S \cdot T$ over a semiring in the \clique model, completing in $O(nz(S)^{1/3} nz(T)^{1/3}/n + 1)$ rounds.
		}
	}
	
	\begin{theorem}
		\label{theorem:theorem1}
		\TheoremMain\footnote{Since we minimize communication rather than element-wise multiplications, the \emph{zero element} does not have to be the zero element of the semiring - any \emph{single} element may be chosen to not be explicitly communicated.}
	\end{theorem}

	An important case of Theorem~\ref{theorem:theorem1}, especially when squaring the adjacency matrix of a graph in order to solve graph problems, is when the sparsities of the input matrices are roughly the same.
	In such a case, Theorem~\ref{theorem:theorem1} gives the following.
	
	\begin{corollary}
		\label{cor:m-non-zero}
		Given two $n \times n$ matrices $S$ and $T$, where $O(nz(S)) = O(nz(T)) = m$,
		Algorithm \alg ~deterministically computes the product $P = S \cdot T$ over a semiring in the \clique model,
		within $O(m^{2/3}/n + 1)$ rounds.
	\end{corollary}
	
	Notice that for $m=O(n^2)$, Corollary~\ref{cor:m-non-zero} gives the same complexity of $O(n^{1/3})$ rounds as given by the semiring multiplication of~\cite{Censor-HillelKK15}.
	
	We apply Algorithm \alg to $4$-cycle counting, obtaining the following.
	\newcommand{\TheoremCounting}
	{
		There is a deterministic algorithm that computes the number of $4$-cycles in an $n$-node graph $G$ in $O(m^{2/3}/n + 1)$ rounds in the \clique model, where $m$ is the number of edges of $G$.
	}
	\begin{theorem}
		\label{theorem:counting}
		\TheoremCounting
	\end{theorem}
	Notice that for $m=O(n^{3/2})$ this establishes $4$-cycle counting in a constant number of rounds.
	
	As described earlier, our algorithm is fast also in the case where only one of the input matrices is sparse, as stated in the following corollary of Theorem~\ref{theorem:theorem1}.
	\begin{corollary}
		\label{cor:one-sparse}
		Given two $n \times n$ matrices $S$ and $T$, where $\min\{O(nz(S)),O(nz(T))\} = m$,
		Algorithm \alg ~deterministically computes the product $P = S \cdot T$  over a semiring in the \clique model,
		within $O((m/n)^{1/3} + 1)$ rounds.
	\end{corollary}
	
	This allows us to compute powers that are larger than 2 of a sparse input matrix. Although we cannot enjoy the guarantees of our algorithm when repeatedly squaring a matrix, because this may require multiplying dense matrices, we can still repeatedly increase its power by 1. This gives the following for computing APSP, whose comparison to the state-of-the-art depends on trade-off between the number of edges in the graph and its diameter.
	
	\newcommand{\TheoremAPSP}
	{
		There is a deterministic algorithm that computes unweighted undirected APSP in an $n$-node graph $G$ in $O(D((m/n)^{1/3}+1))$ rounds in the \clique model, where $m$ is the number of edges of $G$ and $D$ is its diameter.
	}
	\begin{theorem}
		\label{theorem:APSP}
		\TheoremAPSP
	\end{theorem}
	For comparison, the previously known best complexity of unweighted undirected APSP is $O(n^{1-2/\omega})$, given by\cite{Censor-HillelKK15,LeGall16}, which is currently known to be bounded by $O(n^{0.158})$. For a graph with a number of edges that is $m = o(n^{4-6/\omega}/D^3)$, which is currently $o(n^{1.474}/D^3)$, our algorithm improves upon the latter.
	
		Lastly, we leverage the routing techniques developed in our sparse matrix multiplication algorithm in order to introduce an algorithm for the triangle listing problem in the \clique model.
	\newcommand{\TheoremListing}
	{
		There is a deterministic algorithm for triangle listing in an $n$-node, $m$-edge graph $G$  in $O(m/n^{5/3}+1)$ rounds in the \clique model.
	}
	\begin{theorem}
		\label{theorem:triangleListing}
		\TheoremListing
	\end{theorem}
	
	For comparison, two deterministic algorithms by Dolev et al.~\cite{DolevLP12} take $\tilde{O}(n^{1/3})$ and $O(\ceil{\Delta^2/n})$ rounds, while the sparsity-aware randomized algorithm of Pandurangan et al.~\cite{Pandurangan0S16} completes in $\tilde{O}(m/n^{5/3})$, w.h.p. Notice that for general graphs, our algorithm matches the lower bound of $\tilde{\Omega}(n^{1/3})$ by\cite{Pandurangan0S16,IzumiG17}. Additionally, our algorithm for triangle listing implies a triangle counting algorithm. A triangle counting algorithm whose complexity depends on the arboricity $A$ of the graph is given in~\cite{DolevLP12}. Their algorithm completes in $O(A^2 / n + \log_{2+n/A^2} {n})$ rounds. Since $A \geq m/n$, this gives a complexity of $\Omega(m^2/n^3)$, upon which our algorithm provides more than a cubic improvement. The previously known best complexity of triangle and $4$-cycle counting in general graphs is $O(n^{1-2/\omega})$, given by\cite{Censor-HillelKK15}, which is currently known to be bounded by $O(n^{0.158})$. For a graph with a number of edges that is $m = o(n^{8/3-2/\omega})$, which is currently $o(n^{1.824})$, our algorithm improves upon the latter.


\textbf{Roadmap:} The remainder of this section contains an intuitive discussion of the challenges and how we overcome them, followed by a survey of related work and the required preliminaries. Section~\ref{section:SparseMM} gives our sparse matrix multiplication algorithm, and Section~\ref{sec:applications} shows its immediate applications in APSP and 4-cycle counting. Section~\ref{sec:triangle-listing} provides our triangle listing algorithm. We conclude with a discussion in Section~\ref{sec:discussion}.

	\subsection{Challenges and Our Techniques}
	\label{subsec:technique}
	Given two $n\times n$ matrices $S$ and $T$, denote their product by $P=S\cdot T$, for which $P[i][j]=\sum_{k=1}^n S[i][k]T[k][j]$. A common way of illustrating the multiplication is by a \emph{3-dimensional cube} of size $n\times n \times n$, in which the entry $(i,j,k)$ corresponds to the element-wise product $S[i][k]T[k][j]$. In other words, two dimensions of the cube correspond to the matrices $S$ and $T$, and the third dimension corresponds to element-wise products. Each index of the third dimension is a \emph{page}, and $P$ corresponds to the element-wise summation of all $n$ pages.
	
	In essence, the task of distributed matrix multiplication is to assign each of the $n^3$ element-wise multiplications to the nodes of the network, in a way which minimizes the amount of communication that is required.\footnote{We consider all $n^3$ element-wise multiplications rather than Strassen-like algorithms since we work over a semiring and not a ring.} This motivates the goal of assigning the element-wise products to the nodes in a way that balances the number of non-zero elements in $S$ and $T$ that need to be communicated among the nodes, as this is the key ingredient towards minimizing the number of communication rounds. The main obstacle is that a sparse input matrix may be unbalanced, leading to the existence of nodes whose element-wise multiplication operation assignment requires them to obtain many nonzero elements of the input matrices that originally reside in other nodes, and thus necessitating much communication.
	
	As we elaborate upon in Section~\ref{subsec:related}, algorithms for the parallel settings, which encounter the same hurdle, typically first permute the rows and columns of the input matrices in an attempt to balance the structure of the non-zero entries. Ballard et al.~\cite{BallardBDGLST13} write: ``While a priori knowledge of sparsity structure can certainly reduce communication for many important classes of inputs, we are not aware of any algorithms that dynamically determine and efficiently exploit the structure of general input matrices. In fact, a common technique of current library implementations is to randomly permute rows and columns of the input matrices in an attempt to destroy their structure and improve computational load balance."
	
	Our high-level approach, which is \emph{deterministic}, is threefold. The first ingredient is splitting the $n \times n \times n$ cube into $n$ equally sized sub-cubes whose dimensions are determined dynamically, based on the sparsity of the input matrices. The second is indeed permuting the input matrices $S$ and $T$ into two matrices $S'$ and $T'$, respectively. We do so in a subtle manner, for which the resulting matrices exhibit some nice balancing property.\footnote{Note, we do not assume that balancing the distribution of non-zero elements gives a balanced local computation. Our balancing is done for the amount of communication: we assign the amount and identity of matrix entries that should be sent and received by each node in a way that will balance the communication, not necessarily the local computation.} The third ingredient is the innovative part of our algorithm, which assigns the computation of pages of different sub-matrices across the nodes in a \emph{non-consecutive} manner. We elaborate below about these key ingredients, with the aid of Figure~\ref{fig:cube}.
	
	\begin{figure}
		\begin{center}
			\includegraphics[scale=0.5,trim={0 7cm 0 4cm},clip]{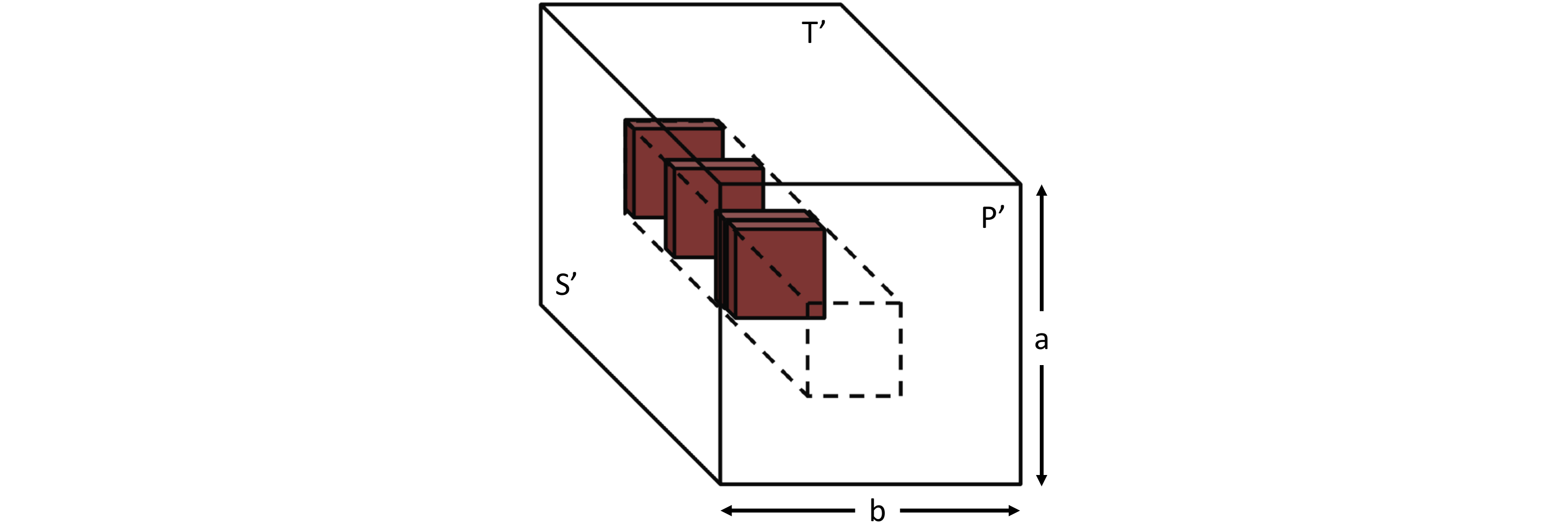}
		\end{center}
		\caption{An illustration of the multiplication cube for $P'=S'T'$. Each sub-matrix is assigned to $n/ab$ nodes, with a not necessarily consecutive page assignment that is computed on-the-fly to minimize communication.}
		\label{fig:cube}
	\end{figure}
	
	\textbf{Permuting the input matrices:}
	We employ standard parallelization of the task of computing the product matrix $P$, by partitioning $P$ into $ab$ equal sized $n/a \times n/b$ sub-matrices denoted by $P_{i,j}$ for $i \in [a], j \in [b]$, and assigning $n/ab$ nodes for computing each sub-matrix.
	
	To this end, we leverage the simple observation that the multiplication of permutations of the rows of $S$ and the columns of $T$ results in a permutation of the product of $S$ and $T$, which can be easily inverted. This observation underlies the first part of our algorithm, in which the nodes permute the input matrices, such that the number of non-zero entries from $S$ and $T$ that are required for computing each $n/a \times n/b$ sub-matrix are roughly the same across the $a \cdot b$ sub-matrices. We call the two matrices, $S'$ and $T'$, that result from the permutations, \emph{\nice~matrices with respect to $(a,b)$}. The rest of our algorithm deals with computing the product of two such matrices.
	This part inherently includes a computation of the best choice for $a$ and $b$ for minimizing the communication.
	
	
	\textbf{Assigning pages to nodes:}
	To obtain each sub-matrix $P_{i,j}$, there are $n$ sub-pages $P_{i,j,\ell}$ which need to be computed and summed. For each $P_{i,j}$, this task is assigned to distinct $n/ab$ nodes, each of which computes some of the $n$ sub-pages $P_{i,j,\ell}$ and sums them locally. The local sums are then aggregated and summed, for obtaining $P_{i,j}$. We utilize the commutativity and associativity properties of summation over the semiring in order to assign sub-pages to nodes in a \emph{non-consecutive} manner, such that the nodes require receiving a roughly equal number of non-zero entries in order to compute their assigned sub-pages.
	
	\textbf{Assigning non-zero matrix entries to nodes:}
	For fast communication in the \clique model using Lenzen's routing scheme (see Section \ref{subsec:preliminaries}), it is moreover paramount that the nodes also \emph{send} a roughly equal amount of non-zero matrix entries. However, it may be the case that a certain row, held by a node $v$, contains a significantly larger number of non-zero entries as compared with other rows. Therefore, we rearrange the entries held by each node such that every node holds a roughly equal amount of non-zero entries that need to be sent to other nodes for computing the $n^3$ products. Notice that in this step we do not rearrange the rows or columns of $S$ or $T$, rather, we redistribute the entries of $S$ and $T$. Thus, a node may hold values which originate from different rows.
	
	\textbf{Routing non-zero elements: }
	Crucially, the assignments made above, for addressing the need to balance sending and receiving, are not global knowledge. That is, for every $P_{i,j}$, the corresponding $n/ab$ nodes decide which matrix entries are received by which node, but this is unknown to the other nodes, who need to send this information. Likewise, the redistribution of entries of $S$ and $T$ across the nodes is not known to all nodes. Nonetheless, clearly, a node must know the destination of each message it needs to send. As a consequence, we ultimately face the challenge of communicating some of this local knowledge. In our solution, a node that needs to receive information from a certain column of $S$ (or row of $T$) sends a request to the nodes holding subsequences of that column (or row) without knowing the exact partition into subsequences. The nodes then deliver the non-zero entries of this column (or row), which allow computing the required element-wise multiplications.
	~\\Our solutions to the three challenges described above, for \emph{sending} and \emph{receiving} as small as possible amounts of information and for resolving a corresponding \emph{routing}, and their combination, are the main innovation of our algorithm.

	\subsection{Related work}
	\label{subsec:related}
	\textbf{Matrix multiplication in the \clique model:}
	A randomized Boolean matrix multiplication algorithm was given by Drucker et al.~\cite{DruckerKO13}, completing in $O(n^{\omega-2})$ rounds, where $\omega$ is the exponent of sequential matrix multiplication. The best currently known upper bound is $\omega < 2.3728639$ \cite{LeGall14a}, implying $O(n^{0.372})$ rounds for the above.
	Later, Censor-Hillel et al.~\cite{Censor-HillelKK15} gave a deterministic algorithm for (general) matrix multiplication over semirings, completing in $O(n^{1/3})$ rounds, and a deterministic algorithm for (general) matrix multiplication over rings, completing in $O(n^{1-2/\omega})$ rounds, which by the current known upper bound on $\omega$ is $O(n^{0.158})$. The latter is a Strassen-like algorithm, exploiting known schemes for computing the product of two matrices over a ring without directly computing all $n^3$ element-wise multiplications. Then, Le Gall~\cite{LeGall16} provided fast algorithms for multiplying rectangular matrices and algorithms for computing multiple instances of products of independent matrices.
	
	\textbf{Related graph computations in the \clique model:}
	Triangle counting in the \clique model was addressed by Dolev et al.~\cite{DolevLP12}, who provided a deterministic $\tilde{O}(n^{d-2/d})$-round algorithm for counting the number of appearances of any $d$-node subgraph, giving triangle counting in $\tilde{O}(n^{1/3})$ rounds. To speed up the computation for sparse instances, ~\cite{DolevLP12} show that every node in a graph with a maximum degree of $\Delta$ can learn its 2-hop neighborhood within $O(\Delta^2/n)$ rounds, implying the same round complexity for triangle counting. They also showed a deterministic triangle counting algorithm completing in $\tilde{O}(A^2/n+\log_{2+n/A^2}{n})$ rounds, where $A$ is the arboricity of the input graph, i.e., the minimal number of forests into which the set of edges can be decomposed. Note that a graph with arboricity $A$ has at most $An$ edges, but there are graphs with arboricity $A$ and a significantly smaller number of edges. Since it holds that $A \geq m/n$, this implies a complexity of $\Omega(m^2/n^3)$ for their triangle counting algorithm, upon which our $O(m^{2/3}/n + 1)$-round algorithm provides a cubic improvement. The deterministic matrix multiplication algorithm over rings of~\cite{Censor-HillelKK15} directly gives a triangle counting algorithm with $O(n^{1-2/\omega})$ rounds.

	For $4$-cycle counting, the algorithm of~\cite{DolevLP12} completes in $\tilde{O}(n^{1/2})$ rounds, and the matrix multiplication algorithm of~\cite{Censor-HillelKK15} implies a solution in $O(n^{1-2/\omega})$ rounds.

	For APSP, the matrix multiplication algorithms of~\cite{Censor-HillelKK15} give $O(n^{1-2/\omega})$ for the unweighted undirected case.
	For weighted directed APSP, $\tilde{O}(n^{1/3})$ rounds are given in~\cite{Censor-HillelKK15}, and
	improved algorithms for weighted (directed and undirected) APSP are given in~\cite{LeGall16}. We mention that our technique could allow for computing weighted APSP, but the cost would be too large due to our iterative multiplication (as opposed to the previous algorithms that can afford iterative \emph{squaring}).
	Algorithms for approximations of APSP are given in~\cite{Nanongkai14, Censor-HillelKK15,LeGall16}.

	Note that for all graph problems, Lenzen's routing scheme~\cite{Lenzen13} (see Section~\ref{subsec:preliminaries}) implies that every node can learn the entire structure of $G$ within $O(m/n)$ rounds, where $m$ is the number of edges (this can also be obtained by a simpler scheme).
	
	\textbf{Sequential matrix multiplication:}
The works of Gustavson~\cite{Gustavson78} and of Yuster and Zwick~\cite{YusterZ05} give matrix multiplication algorithms that are faster than $O(n^\omega)$, for sparse matrices. In the latter, the exact complexity depends also on the exponents of certain rectangular matrix multiplications. In a nutshell, the latter algorithm cleverly splits the input matrices into two sets of rows and columns, one dense and one very sparse, by balancing the complexities of multiplying each part. This algorithm is designed to reduce the number of multiplication operations, which is not a direct concern for distributed algorithms, in which the main cost is due to communication. Le Gall~\cite{LeGall12} improved this result for some range of sparsity by improving general rectangular matrix multiplication, for which a further improvement was recently given by Le Gall and Urrutia~\cite{LeGallU18}. Kaplan et al.~\cite{KaplanSV06} give an algorithm for multiplying sparse rectangular matrices, and Amossen and Pagh~\cite{AmossenP09} give a fast algorithm for the case of sparse square matrices for which the product is also sparse.

\textbf{Parallel matrix multiplication:}
There are many known matrix multiplication algorithms in parallel models of computing, of which we give a non-exhaustive overview here. These algorithms are typically categorized as 1D, 2D or 3D algorithms, according to the manner in which the element-wise products are split across the nodes (by pages, rectangular prisms, or sub-cubes, respectively). Algorithms are also distinguished according to whether they are \emph{sparsity-dependent}, that is, whether they leverage the structure of the non-zero elements rather than their number only.

For example, the work of Ballard et al.~\cite{BallardDKS16} looks for a good assignment to nodes by modeling the problem as a hypergraph. Randomly permuting the rows or columns of the input matrices can be expected to result in a balanced structure of nonzero elements. Examples for algorithms relying on \emph{random permutations} of the input matrices can be found in~\cite{BulucG11,BulucG12}.

Additional study appears in Solomonik et al.~\cite{SolomonikBVH17} and in Azad et al.~\cite{AzadBBDGSTW16}. The latter proposes algorithms of various types of dimensionality and also employ permutations on the input matrices.

In a similar spirit to random permutations, fast algorithms can be devised for \emph{random matrices}, as shown by the work of Ballard et al.~\cite{BallardBDGLST13}, which provides fast matrix multiplication algorithms for matrices representing sparse random Erdos-Renyi graphs. The random positioning of the nonzero elements gives rise to the analysis of the complexity of their algorithm.

Solomon and Demmel~\cite{SolomonikD11} give a 2.5-dimensional matrix multiplication algorithm, in the sense that the cube is split to $(n/c)^{1/2} \times (n/c)^{1/2} \times c$ sub-cubes. A recent work by Lazzaro et al.~\cite{LazzaroVHS17} provides a 2.5D algorithm that is also suitable for sparse matrices, which also employs random permutations of rows and columns.

We note that the sparse-dense parallel multiplication algorithm of Tiskin~\cite{Tiskin01} shuffles a single matrix in a sparsity-aware manner, but it does so to only \emph{one} of the two multiplied matrices. This work also contains a path-doubling technique for computing APSP, but it is not clear what improvement this would constitute in the complexity measures addressed in our paper.
An additional algorithm for sparse-dense multiplication is given in Koanantakool et al.~\cite{KoanantakoolABM16}.

Finally, we note that experimental studies appear in some of the above papers, and in additional works, such as by Ahmed et al.~\cite{AhmedHHRP17} and by Deveci et al.~\cite{DeveciTR18}.

In comparison to all of the above, our algorithm is a sparsity-dependent 3D algorithm. Yet, our algorithm presents the additional complication of determining the dimensions of the sub-cube assigned to each node \emph{dynamically}, depending on the number of non-zero elements in each input matrix. Moreover, the sub-cubes may be of \emph{non-consecutive} pages. Further, our transformations on the rows and columns are \emph{deterministic} and are applied to \emph{both} matrices. Being able to permute the two matrices and assign pages to nodes in ways which leverage the sparsity of both matrices is the crux of the novelty of our algorithm.

	\subsection{Preliminaries}
	\label{subsec:preliminaries}
	\textbf{Model:}
	The \clique model consists of a set $[n]=\{1,\dots,n\}$ of nodes in a fully connected synchronous network, limited by a bandwidth of $O(\log{n})$ bits per message.
	
	In an instance of multiplication of two matrices $S$ and $T$, the input to each node $v$ is row $v$ of each matrix and its output should be row $v$ of $P = S\cdot T$. For a graph problem over a graph $G$ of $n$ nodes, we identify the nodes of the \clique model with the nodes of $G$, and the input to node $v$ in the \clique model is its input in $G$.
	
	As defined earlier, for a matrix $A$ we denote by $nz(A)$ the number of non-zero elements of $A$. Throughout the paper, we also need to refer to the number of non-zero elements in certain sub-matrices or sequences. We will therefore overload this notation, and use $nz(X)$ to denote the number of non-zero elements in any object $X$.

	A pair of integers $(a,b)$ is \emph{$n$-split} if $a,b \in [n]$, both $a$ and $b$ divide $n$, and $n/ab \geq 1$. The requirement that $a$ and $b$ divide $n$ is for simplification only and could be omitted. Eventually, the $n$-split pair that will be chosen is $a=n \cdot nz(S)^{1/3}/nz(T)^{2/3}$ and $b=n \cdot nz(T)^{1/3}/nz(S)^{2/3}$.
	
	For a given $n$-split pair $(a,b)$, it will be helpful to associate each node $v \in [n]$ with three indices, two indicating the $P_{i,j}$ sub-matrix to which the node is assigned, and one distinguishing it from the other nodes assigned to $P_{i,j}$. Hence, we denote each node $v$ also as $v_{i,j,k}$, where $i \in [a]$, $j \in [b]$, and $k \in [n/ab]$. The assignment of indices to the nodes can be any arbitrary one-to-one function from $[n]$ to $[a] \times [b] \times [n/ab]$.

	Throughout our algorithms, we implicitly comply with the following: (I) no information is sent for matrix entries whose value is zero, and (II) when the value of a non-zero entry is sent, it is sent alongside its location in the matrix. Since sending the location within an $n \times n$ matrix requires $O(\log{n})$ bits, the overhead in the complexity is constant.
	
	\textbf{Lenzen's routing scheme:}
	A useful tool in designing algorithms for the \clique model is \emph{Lenzen's routing scheme}~\cite{Lenzen13}. In this scheme, each of the $n$ nodes can send and receive $n-1$ messages (of $O(\log{n})$ bits each) in a constant number of rounds. While this is simple to see for the simplest case where each node sends a single message to every other node, the power of Lenzen's scheme is that it applies to any (multi)set of source-destination pairs, as long as each node is source of at most $n-1$ messages and destination of at most $n-1$ messages. Moreover, the multiset of pairs does not need to be known to all nodes in advance, rather each sender only needs to know the recipient of its messages. Employing this scheme is what underlies our incentive for balancing the number of messages that need to be sent and received by all the nodes.
	
	\textbf{Useful combinatorial claims:}
	The following are simple combinatorial claims that we use for routing messages in a load-balanced manner.
	\begin{claim}
		\label{combi:SimplePartition}
		Let $A=(a_1, \dots, a_t)$ be a finite set and let $1 \leq c \leq t$ be an integer. There exists a partition of $A$ into $\ceil{t/c}$ subsets of size at most $c+1$ each.
	\end{claim}
	\begin{proofof}{Claim~\ref{combi:SimplePartition}}
	Let $A_j=\{a_{(j-1)(c+1)+1},\dots a_{j(c+1)}\}$ for every $1\leq j \leq \ceil{t/c}$, where if $j>t$ then we ignore the notation $a_j$. It is easy to verify that each set is of size at most $c+1$, and that $A=\cup_{1\leq j \leq \ceil{t/c}}{A_j}$.
\end{proofof}

	\begin{claim}
		\label{combi:SimpleByAvg}
		Let $A^i=(a^i_1, \dots, a^i_{t_i})$ be a finite set, for $1 \leq i \leq n$. Let $avg = (\sum_{1 \leq i \leq n}{t_i})/n$.
		There exists a partition of each $A^i$ into $\ceil{t_i/avg}$ subsets of size at most $avg+1$ each, such that the total number of subsets is at most $2n$.
	\end{claim}
	\begin{proofof}{Claim~\ref{combi:SimpleByAvg}}
	By Claim~\ref{combi:SimplePartition} with $c=avg$, for every $1 \leq i \leq n$, there is a partition of $A^i$ into $\ceil{t_i/avg}$ subsets of size at most $avg+1$. The total number of subsets is $\sum_{1 \leq i \leq n}{\ceil{t_i/avg}} \leq \sum_{1 \leq i \leq n}{t_i/avg+1} \leq n+ (\sum_{1 \leq i \leq n}{t_i})/avg \leq 2n$, as required.
\end{proofof}

	\begin{claim}
		\label{combi:partitioning}
		Given a sorted finite multiset $A=(a_1, \dots, a_n)$ of natural numbers, an integer $x\in \mathbb{N}$ such that for all $i\in [n]$ it holds that $a_i\leq x$, and an integer $k$ that divides $n$, there exists a partition $A=\cup_{i=j}^{k}{A_j}$ into $k$ multisets $A_j$, $1\leq j \leq k$, of equal size $n/k$, such that for all $1\leq j\leq k$  it holds that $sum(A_j)\leq sum(A)/k+x$.
	\end{claim}
\begin{proofof}{Claim~\ref{combi:partitioning}}
	We show that $A_j=\{a_{j+\ell k}\mid 0\leq \ell < n/k\}$ gives the claimed partition.
	Since $A$ is sorted, we have that $sum(A_j)=\sum_{\ell=0}^{n/k-1}{a_{j+\ell k}} \leq \sum_{\ell=0}^{n/k-1}{a_{k+\ell k}} = sum(A_k)$, for every $j \in [k]$.
	In addition, removing the last element from $A_k$ gives that $sum(A_k)-a_n = \sum_{\ell=0}^{n/k-2}{a_{k+\ell k}} \leq  \sum_{\ell=0}^{n/k-2}{a_{j+(\ell+1)k}} \leq sum(A_j)$, for every $j \in [k]$.
	This implies that $k (sum(A_k)-a_n) \leq \sum_{j=1}^k sum(A_j) = sum(A)$.
	Since $a_n\leq x$, we conclude that
	$sum(A_j) \leq sum(A_k) =  sum(A_k) - a_n + a_n = \frac{k (sum(A_k)-a_n)}{k}+a_n \leq sum(A)/k+x$, for every $j \in [k]$, which completes the proof.
\end{proofof}

	\section{Fast Sparse Matrix Multiplication}
	\label{section:SparseMM}
	Our main result is Theorem \ref{theorem:theorem1}, stating the guarantees of our principal algorithm \alg (Sparse Matrix Multiplication) for fast multiplication of sparse matrices. Algorithm \alg first manipulates the structure of its input matrices and then calls algorithm \algnice (Sparse Balanced Matrix Multiplication), which solves the problem of fast sparse matrix multiplication under additional assumptions on the distributions of non-zero elements in the input matrices, which are defined next. In Section \ref{subsec:GivenSparse}, we show how \alg computes general matrix multiplication $P=ST$, given Algorithm \algnice and Theorem~\ref{theorem:MultiplyGivenNiceST}. Algorithm \algnice, and Theorem~\ref{theorem:MultiplyGivenNiceST} which states its guarantees, are deferred to Section~\ref{subsec:SBMM}.
	
	\begin{theorem-repeat}{theorem:theorem1}
		\TheoremMain
	\end{theorem-repeat}
	
	We proceed to presenting Theorem \ref{theorem:MultiplyGivenNiceST} which discusses \algnice. \algnice multiplies matrices $S'$ and $T'$ in which the non-zero elements are roughly balanced between portions of the rows of $S'$ and columns of $T'$.
	In what follows, for a matrix $A$, the notation $A[x:y][*]$ refers to rows $x$ through $y$ of $A$ and the notation $A[*][x:y]$ refers to columns $x$ through $y$ of $A$.
	In the following definition we capture the needed properties of well-balanced matrices.
	\begin{definition}
		\label{def:nice-matrices}
		Let $S$ and $T$ be $n \times n$ matrices and let $(a, b)$ be an $n$-split pair.
		For every $i \in [a]$ and $j \in [b]$, denote $S_i = S[\SIndecies{i}][*]$ and $T_j = T[*][\TIndecies{j}]$.
		We say that $S$ and $T$ are a \emph{\nice ~pair of matrices with respect to $(a,b)$}, if:
		\begin{itemize}
			\item{$S$-condition:} For every $i \in [a]$, $nz(S_i) \leq nz(S)/a + n$.
			\item{$T$-condition:} For every $j \in [b]$, $nz(T_j) \leq nz(T)/b + n$.
		\end{itemize}
	\end{definition}
	
	These conditions ensure that \emph{bands} of adjacent rows of $S$ and columns of $T$ contain roughly the same number of non-zero elements. We can now state our theorem for multiplying \nice ~matrices, which summarizes our algorithm \algnice.
	\begin{theorem}
		\label{theorem:MultiplyGivenNiceST}
		Given two $n \times n$ matrices $S$ and $T$ and an $n$-split pair $(a, b)$, if $S$ and $T$ are a \nice~ pair with respect to $(a, b)$, then Algorithm \algnice ~deterministically computes the product $P = S \cdot T$ over a semiring in the \clique model, completing in
		$O(nz(S)\cdot b/n^2+nz(T)\cdot a/n^2+n/ab + 1)$ rounds.
	\end{theorem}
	
	We show that $O(1)$ rounds are sufficient in the \clique for transforming any two general matrices $S$ and $T$ to \nice\  matrices $S'$ and $T'$ by invoking standard matrix permutation operations. Therefore, in essence, Algorithm \alg performs permutation operations on $S$ and $T$, generating the matrices $S'$ and $T'$, respectively, invokes \algnice on $S'$ and $T'$ to compute $P' = S'T'$, and finally recovers $P$ from $P'$.

	\subsection{Fast General Sparse Matrix Multiplication - Algorithm \alg}
	\label{subsec:GivenSparse}
	\textbf{Algorithm Description:} First, each node distributes the entries in its row of $T$ to other nodes in order for each node to obtain its column in $T$. Then, the nodes broadcast the number of non-zero elements in their respective row of $S$ and column of $T$, in order for all nodes to compute $nz(S)$ and $nz(T)$. Having this information, the nodes locally compute the $n$-split pair $(a,b)$ that minimizes the expression $nz(S)\cdot b/n^2+nz(T)\cdot a/n^2+n/ab$, which describes the round complexities of each of the three parts of  Algorithm \algnice. It can be shown that the pair $(n\cdot nz(S)^{1/3} / nz(T)^{2/3},  n\cdot nz(T)^{1/3} / nz(S)^{2/3})$ minimizes this expression. Then, the nodes permute the rows of $S$ and columns of $T$ so as to produce matrices $S'$ and $T'$ which have the required balance. Subsequently, Algorithm \algnice is executed on the permuted matrices $S'$ and $T'$, followed by invoking the inverse permutations on the product $P'=S'T'$ in order to obtain the product $P = S \cdot T$ of the original matrices. A pseudocode of \alg is given in Algorithm~\ref{procedure:reduction1}.
	
	\begin{algorithm}
		\ForEach{$u\in [n]$, $u\neq v$}
		{\label{line:one}
			\Send $T[v][u]$ to node $u$
		}
		
		\ForEach{$u\in [n]$, $u\neq v$}
		{
			\Send $nz(S[v][*])$ to node $u$\\
			\Send $nz(T[*][v])$ to node $u$
		}
		
		$nz(S)\leftarrow \sum_{u \in [n]}{nz(S[u][*])}$\\
		$nz(T)\leftarrow \sum_{u \in [n]}{nz(T[*][u])}$
		
		$(a,b) \leftarrow \argmin_{n\text{-split pairs } (a, b)}{\{nz(S)\cdot b/n^2+nz(T)\cdot a/n^2+n/ab\}}$ \label{abCalculating}
		
		Let ${A^S_1,\dots, A^S_a}$ be the partition of the sorted multiset of $\{nz(S[u][*]) | u \in [n]\}$, into $a$ multisets with a bound $x = n$ on its elements, and let ${A^T_1,\dots, A^T_b}$ be the partition of the sorted multiset of $\{nz(T[*][u]) | u \in [n]\}$, into $b$ multisets with a bound $x = n$ on its elements, both proven to exist in Claim~\ref{combi:partitioning}. \label{line:partition}
		
		Let $\sigma$ be a permutation for which its $n \times n$ permutation matrix $A_{\sigma}$ is such that the rows of the matrix $S'=A_{\sigma} S$ that correspond to any single $A^S_u$ are adjacent, and let $\tau$ be a permutation for which its $n \times n$ permutation matrix $A_{\tau}$ is such that the columns of the matrix $T'=T A_{\tau}$ that correspond to any single $A^T_u$ are adjacent. \label{S'T'Calculating}
		
		\Send $S[v][*]$ to node $\sigma(v)$
		
		\Send $T[*][v]$ to node $\tau(v)$ \label{beforeFunctionCall}
		
		\ForEach{$u\in [n]$, $u\neq v$}
		{\label{line:col-to-row}
			\Send $T'[u][v]$ to node $u$
		}

		$P' \leftarrow \algnice(S',T',a,b)$\label{line:algnice}
		
		\ForEach{$u\in [n]$, $u\neq v$}
		{
			\Send $P'[\sigma^{-1}(v)][\tau^{-1}(u)]$ to node $u$
		}
		\caption{\textbf{\alg($S,T$)}: Computing the product $P=S\cdot T$. Code for node $v \in \{1,\dots,n\}$.}
		\label{procedure:reduction1}
	\end{algorithm}

	\begin{proofof}{Theorem~\ref{theorem:theorem1}}
		To prove correctness, we need to show that the matrices $S'$ and $T'$ computed in Line~\ref{S'T'Calculating} are a \nice ~pair of matrices with respect to the $n$-split pair $(a,b)$ that is determined in Line~\ref{abCalculating}. Once this is proven, the correctness of the algorithm is as follows. In Lines \ref{line:one}-\ref{beforeFunctionCall} the matrices $S'$ and $T'$ are computed and are distributed among the nodes such that each node $v \in [n]$ holds row $v$ of $S$ and column $v$ of $T$. The loop of Line~\ref{line:col-to-row} is only for consistency, having the input to \algnice be the respective rows of both $S'$ and $T'$.
		Assuming the correctness of algorithm \algnice given in Theorem \ref{theorem:MultiplyGivenNiceST}, the matrix $P'$ computed in Line~\ref{line:algnice} is the product $P' = S'T'$.
		Finally, in the last loop, node $v$ receives row $v$ of $P=A^{-1}_{\sigma} P' A^{-1}_{\tau}$, completing the correctness of the Algorithm \alg.
		
		We now show that $S'$ and $T'$ are indeed a \nice ~pair of matrices with respect to $(a,b)$. To this end, we first need to show that for all $i \in [a]$, the number of non-zero elements in $S'_i$ is at most $nz(S')/a + n$. By construction, the number of non-zero elements in $S'_i=S'[(i-1)(n/a)+1:i (n/a)][*]$ is exactly $sum(A^{S}_i)$ of the partition computed in Line~\ref{line:partition}. By Claim~\ref{combi:partitioning} this is bounded by $sum(A)/k+x$, which in our case is $nz(S)/a+n=nz(S')/a+n$. Thus, $S'$ satisfies the $S$-condition of Definition~\ref{def:nice-matrices}. A similar argument shows that $T'$ satisfies the $T$-condition of Definition~\ref{def:nice-matrices}.
		
		For the complexity, we sum the number of rounds as follows. The first loop allows every node $v$ to obtain column $v$ of $T$, while in the second loop the nodes exchange the sums of non-zero elements in rows and columns of $S$ and $T$, respectively. Even without the need to resort to Lenzen's routing scheme, both of these loops can be completed within $O(1)$ rounds. A similar argument shows that $O(1)$ rounds suffice for permuting $S$ and $T$ into $S'$ and $T'$, and for permuting $P'$ back into $P$. Thus, all lines of the pseudocode excluding Line~\ref{line:algnice} complete in $O(1)$ rounds. This implies that the complexity of Algorithm \alg ~equals that of Algorithm \algnice when given $S'$, $T'$, $a$, and $b$ as input. By Theorem~\ref{theorem:MultiplyGivenNiceST} and due to the choice of $a$ and $b$ in Line~\ref{abCalculating}, this complexity is $O(\min_{n\text{-split pairs } (a,b)} \{nz(S)\cdot b/n^2+nz(T) \cdot a/n^2+n/ab + 1\})$. Choosing $a=n\cdot nz(S)^{1/3} / nz(T)^{2/3}$ and $b= n\cdot nz(T)^{1/3} / nz(S)^{2/3}$ gives a complexity of $O(nz(S)^{1/3}nz(T)^{1/3}/n + 1)$ rounds, which can be shown to be optimal.
	\end{proofof}

	\subsection{Fast Sparse Balanced Matrix Multiplication - Algorithm \algnice}
	\label{subsec:SBMM}
	Here we present \algnice and prove Theorem \ref{theorem:MultiplyGivenNiceST}. We begin with a short overview of the algebraic computations and node allocation in \algnice. We then proceed to presenting a communication scheme detailing how to perform the computations of \algnice in the \clique model in $O(M_S\cdot b/n^2+M_T\cdot a/n^2+n/ab + 1)$ rounds of communication.
	
	\textbf{Algorithm Description:} Consider the partition of $P$ into $ab$ rectangles, such that $\forall (i, j) \in [a] \times [b]$, \emph{sub-matrix} $P_{i,j} = P[\SIndecies{i}][\TIndecies{j}]$. Each sub-matrix $P_{i,j}$ is an $n/a \times n/b$ matrix, i.e., has $n^2/ab$ entries. Notice that $P_{i,j} = S_i\cdot T_j$. We assign the computation of $P_{i,j}$ to a unique set of $n/ab$ nodes $N_{i,j} = \{v_{i,j,k} | k \in [n/ab]\}$.
	
	In the initial phase of algorithm \algnice, for every $(i,j) \in [a] \times [b]$, each non-zero element of $S_i$ and $T_j$ is sent to some node in $N_{i,j}$. Due to the \nice~property of $S$ and $T$, all $S_i$'s have roughly the same amount of non-zero elements, and likewise all $T_j$'s. Therefore, each set of nodes $N_{i,j}$ receives roughly the same amount of non-zero elements from $S$ and $T$.

	Within each $N_{i,j}$, the computation of $P_{i,j}$ is carried out according to the following framework. For $\ell \in [n]$, denote each page of $P_{i,j}$ by $P_{i,j,\ell} = S_i[*][\ell] \cdot T_j[\ell][*]$. The computation of the $n$ different $P_{i,j,\ell}$ sub-matrices is split among the nodes in $N_{i,j}$ as follows: The set $[n]$ is partitioned into $A_{i,j,1}, \dots, A_{i,j,n/ab}$ such that for each $k \in [n/ab]$, node $v_{i,j,k} \in N_{i,j}$ is required to compute the entries of the matrices in the set $\{P_{i,j,\ell} | \ell \in A_{i,j,k}\}$. Then, node $v_{i,j,k} \in N_{i,j}$ locally sums its computed sub-matrices to produce $P_{i,j}^k = \sum_{\ell \in A_{i,j,k}}P_{i,j,\ell}$.
	Clearly, due to the associativity and commutativity of the addition operation in the semiring, it holds that $P_{i,j} = \sum_{\ell\in [n]}P_{i,j,\ell} = \sum_{k \in [n/ab]}P_{i,j}^k$. Therefore, once every node $v_{i,j,k}$ has $P_{i,j}^k$, the nodes can collectively compute $P$, and redistribute its entries in a straightforward manner such that each node obtains a distinct row of $P$.

	\textbf{Implementing \algnice:} A pseudocode for Algorithm \algnice~ is given in Algorithm~\ref{procedure:ASBMM}, which consists of three components: exchanging information between the nodes such that every node $v_{i,j,k}$ has the required information for computing $P_{i,j,\ell}$ for every $\ell \in A_{i,j,k}$, local computation of $P_{i,j}^k$ for each $(i,j,k) \in [a] \times [b] \times [n/ab]$ and, finally, the communication of the $P_{i,j}^k$ matrices and assembling of the rows of $P$.
	
	\begin{algorithm}
		\algASBMMSubroutineOne($S,T,a,b$)\label{line:exchange}
		
		Locally compute $P_{i,j,\ell}$ for every $\ell \in A_{i,j,k}$\label{line:local1}
		
		Locally compute $P_{i,j}^k = \sum_{\ell\in A_{i,j,k}}{P_{i,j,\ell}}$\label{line:local2}
		
		\ForEach{$t \in [n/a]$}{\label{line:loopsend}
			\Send $P_{i,j}^k[t][*]$ to node of respective row
		}
		
		\ForEach{$\ell \in [n]$}{\label{line:looplocal}
			$P[v][\ell]\leftarrow $ sum of $n/ab$ respective elements received for this entry
		}
		
		\caption{\textbf{\algnice(S,T,a,b)}: Computing the product $P=ST$, for $S$ and $T$ that are \nice~w.r.t. $(a,b)$. Code for node $v \in \protect{[n]}$, which is also denoted $v_{i,j,k}$.}
		\label{procedure:ASBMM}
	\end{algorithm}
	
	The technical challenge is in Line~\ref{line:exchange}, upon which we elaborate below.
	In Lines~\ref{line:local1}-\ref{line:local2}, only local computations are performed, resulting in each node $v_{i,j,k}$ holding $P_{i,j}^k$.
	In Line~\ref{line:loopsend}, each node sends each row of its sub-matrix $P_{i,j}^k$ to the appropriate node, so that in Line~\ref{line:looplocal} each node can sum this information to produce its row in $P$.
	Formally, we prove the following.
	
	\begin{lemma}
		\label{lemma:afterExchange}
		Lines~\ref{line:local1}-\ref{line:looplocal} of Algorithm~\ref{procedure:ASBMM} complete in $O(n/ab + 1)$ rounds, producing a row of $P=ST$ for every node.
	\end{lemma}
	\begin{proofof}{Lemma~\ref{lemma:afterExchange}}
	Lines~\ref{line:local1}-\ref{line:local2} require no communication.
	
	In the loop of Line~\ref{line:loopsend}, each node sends each of the entries of its sub-matrix $P_{i,j}^k$ to a single receiving node, implying that each node sends $n^2/ab$ messages.
	To verify that this is also the number of messages received by each node, recall that each entry of the matrix $P$ is computed in Line~\ref{line:looplocal} as a summation of $n/ab$ entries, each is an entry of $P_{i,j}^k$ for two appropriate values of $i,j$ and for all $k \in [n/ab]$. Since such $n/ab$ messages need to be received for every entry of the row, this results in receiving $n^2/ab$ messages.
	
	For the above, we use Lenzen's routing scheme, which completes in $n/ab$ rounds, completing the proof.
\end{proofof}

	The remainder of this section is dedicated to presenting and analyzing Line~\ref{line:exchange}.
	During this part of the algorithm, for every $(i,j) \in [a] \times [b]$, each entry in $S_i$ and $T_j$ needs to be sent to a node in $N_{i,j}$. As per our motivation throughout the entire algorithm, we strive to achieve this goal in a way which ensures that all nodes send and receive roughly the same number of messages. This leads to the following three challenges which we need to overcome.
	
	\textbf{Sending Challenge:} Initially, node $v$ holds row $v$ of $S$ and row $v$ of $T$. Every column $v$ of $S$ needs to be sent to $b$ nodes - one node in each $N_{i,j}$ for an appropriate $i \in [a]$ and every $j \in [b]$. Similarly, every row of $T$ needs to be sent to $a$ nodes - one in each $N_{i,j}$ for an appropriate $j \in [b]$ and every $i \in [a]$. If we were to trivially choose node $v$ to send all these messages, then node $v$ would need to send $nz(S[*][v])\cdot b + nz(T[v][*])\cdot a$ messages. Since $nz(S[*][v])$ and $nz(T[v][*])$ may widely vary for different values of $v$, it may be the case that some nodes send a significant amount of messages while others are relatively silent.

	\textbf{Receiving Challenge:} Since $S$ and $T$ are \nice~w.r.t. $(a,b)$, for every $(i,j) \in [a]\times[b]$ it holds that the number of messages to be received by each set of nodes $N_{i,j}$ is at most $nz(S)/a + nz(T)/b + 2n$. This ensures that each node set $N_{i,j}$ receives roughly the same amount of messages as every other node set. The challenge remains to ensure that \emph{within} any given node set $N_{i,j}$, every node receives roughly the same number of messages.

	\textbf{Routing Challenge:} When overcoming the above mentioned challenges in a non-trivial manner, all nodes locally determine that they are senders and recipients of certain messages with the guarantee that each node sends and receives roughly the same number of messages.
	However, these partitions of sending and receiving messages are obtained independently and thus are not global knowledge; a sender of a message \emph{does not necessarily know} who the recipient is. The routing challenge is thus to ensure that each node associates the correct recipient with every message that it sends.
	\subsubsection{\algASBMMSubroutineOne($S,T,a,b$)}
	We next present our implementation of \algASBMMSubroutineOne($S,T,a,b$) which solves the above challenges in an on-the-fly manner. To simplify the presentation, we split \algASBMMSubroutineOne($S,T,a,b$) into its three components, as given in the pseudocode of Algorithm~\ref{procedure:ASBMM1}.

	\begin{algorithm}
		Compute-Sending\\
		Compute-Receiving\\
		Resolve-Routing
		\caption{\textbf{\algASBMMSubroutineOne(S,T,a,b)}: Sending each entry of $S_{i}$, $T_j$ to a node in $N_{i,j}$, for every $(i,j) \in {[}a{]} \times {[}b{]}$.}
		\label{procedure:ASBMM1}
	\end{algorithm}

	\textbf{Compute-Sending:} In Compute-Sending, whose pseudocode is given in Algorithm~\ref{procedure:ASBMM1.1}, we overcome the sending challenge.
	The nodes communicate the distribution of non-zero elements across the columns of $S$ and the rows of $T$ and reorganize the entries held by each node such that all nodes hold roughly the same amount of non-zero elements of $S$ and $T$.
	
	Notably, in order to enable fast communication in Resolve-Routing, Algorithm \ref{procedure:ASBMM1.1} must guarantee no node holds entries of more than two columns of $S$ and two rows of $T$.
	
	\begin{algorithm}
		\ForEach{$u\in [n]$, $u\neq v$}
		{
			\Send $S[u][v]$ to node $u$\label{line:col}
		}
		\ForEach{$u\in [n]$, $u\neq v$}
		{
			\Send $nz(S[*][v])$ to node $u$ \label{line:nzS}\\
			\Send $nz(T[v][*])$ to node $u$ \label{line:nzT}
		}
		
		$avg(S)\leftarrow (\sum_{u \in [n]}{nz(S[*][u])})/n$\\
		$avg(T)\leftarrow (\sum_{u \in [n]}{nz(T[u][*])})/n$\\
		
		Let $S_1^v,\dots,S_{\ceil{nz(S[*][v])/avg(S)}}^v$ be a partition of the non-zero elements of $S[*][v]$ into sets of size at most $avg(S) + 1$ and let $T_1^v,\dots,T_{\ceil{nz(T[v][*])/avg(T)}}^v$ be a partition of the non-zero elements of $T[v][*]$ into sets of size at most $avg(T) + 1$, both proven to exist in Claim~\ref{combi:SimplePartition}. We refer to these sets as \emph{subsequences}.
		
		Assign two subsequences of $S$, denote by $B_S(v)$, and two subsequences of $T$, denote by $B_T(v)$ to each node $v$. For each subsequence $B$, denote by $v(B)$ the node to which $B$ is assigned. \label{procedure:ASBMM1.1BlockAssignmentStep}
		
		\ForEach{$B \in \{S_1^v,\dots,S_{\ceil{nz(S[*][v])/avg(S)}}^v, T_1^v,\dots,T_{\ceil{nz(T[v][*])/avg(T)}}^v\}$}
		{
			\Send $B$ to node $v(B)$
		}
		\caption{\textbf{Compute-Sending}: Code for node $v \in \protect{[n]}$, which is also denoted $v_{i,j,k}$.}
		\label{procedure:ASBMM1.1}
	\end{algorithm}
	
	\begin{lemma}
		\label{lemma:SendingSolution}
		Algorithm~\ref{procedure:ASBMM1.1} completes in $O(1)$ rounds, after which the entries of $S$ and $T$ are evenly redistributed across the nodes such that every node holds elements from at most 2 columns of $S$ and 2 rows of $T$ and such that every node $v$ knows for every node $u$ the indices of the two columns of $S$ and two rows of $T$ from which the elements which $u$ holds are taken.
	\end{lemma}
	
	\begin{proof}
		In Lines~\ref{line:col},~\ref{line:nzS},~\ref{line:nzT} the nodes exchange entries of $S$ such that each node holds a distinct column of $S$, and knows the number of non-zero entries in each column of $S$ and in each row of $T$. This allows local computation of the average number of non-zeros in the following two lines, as well as locally computing the (same) partition into subsequences.
		
		By Claim~\ref{combi:SimpleByAvg}, in total across all $n$ columns there are at most $2n$ subsequences of entries from $S$, and similarly there are at most $2n$ subsequences from $T$. Since $\forall u \in [n]$, all nodes know $nz(S[*][u])$ and $nz(T[u][*])$, then all nodes know how many subsequences are created for each $u$. Thus, all nodes can agree in Line~\ref{procedure:ASBMM1.1BlockAssignmentStep} on the assignment of the subsequences, with each node assigned at most 2 subsequences of entries of $S$ and 2 of entries of $T$. Crucially for what follows, all the nodes know the column $\ell$ in $S$ or the row $\ell$ in $T$ to which the subsequence $B$ belongs. We denote this index $\ell(B)$. The entries of each subsequence $B$ are then sent to its node $v(B)$ in the following loop.
		
		For the round complexity, note that a node $v$ sends a single message to every other node in each of Lines~\ref{line:col},~\ref{line:nzS}, and~\ref{line:nzT}. The rest of the computation until Line~\ref{procedure:ASBMM1.1BlockAssignmentStep} is done locally. Therefore, these lines complete within $3$ rounds.
		
		In the last loop of Algorithm~\ref{procedure:ASBMM1.1}, node $v$ potentially sends all subsequences with entries from column $v$ of $S$ and row $v$ of $T$. Due to the facts that each subsequence is sent only once, no subsequences overlap, and all the subsequences which $v$ send are parts of a single column of $S$ and a single row of $T$, node $v$ sends at most $2n$ messages during this loop. Additionally, since every node receives at most $4$ subsequences and each subsequence consists of most $n$ entries, each node receives at most $4n$ messages. Thus, by using Lenzen's routing scheme, this completes in $O(1)$ rounds as well.
	\end{proof}
	
	\textbf{Compute-Receiving:} The pseudocode for Compute-Receiving is given in Algorithm~\ref{procedure:ASBMM1.2}.
	This algorithm assigns the $P_{i,j,\ell}$ matrices to different nodes in $N_{i,j}$. Specifically, each node $v_{i,j,k}$ in $N_{i,j}$ is assigned $ab$ such matrices, while verifying that all nodes in $N_{i,j}$ require roughly the same amount of non-zero entries from $S$ and $T$ in order to compute all their assigned $P_{i,j,\ell}$ matrices. Since each sub-matrix $P_{i,j,\ell}$ is defined as $P_{i,j,\ell}= S_i[*][\ell] \cdot T_j[\ell][*]$, we define the communication cost of computing $P_{i,j,\ell}$ to be $w(P_{i,j,\ell}) = nz(S_i[*][\ell]) + nz(T_j[\ell][*])$. By this definition, in order to obtain that each node in $N_{i,j}$ requires roughly the same amount of messages in order to compute all of its assigned $P_{i,j,\ell}$, we assign the $P_{i,j,\ell}$ matrices to the nodes of $N_{i,j}$ such that the total communication cost, as measured by $w$, of all matrices assigned to a given node is roughly the same for all nodes.
	
	\begin{algorithm}
		\ForEach{$N_{i',j'}, i',j' \in [a] \times [b]$}
		{\label{line:BigLoop}
			\ForEach{$u \in N_{i',j'}$}
			{
				\ForEach{$B \in B_S(v)$}{
					\Send $nz(S_{i'} \cap B)$ to node $u$
				}
				\ForEach{$B \in B_T(v)$}{
					\Send $nz(T_{j'} \cap B)$ to node $u$
				}
			}
		}
		
		\ForEach{$\ell \in [n]$}
		{\label{line:wLoop}
			$w(P_{i,j,\ell}) \leftarrow nz(S_i[*][\ell]) + nz(T_j[\ell][*])$
		}
		
		Let $A_{i,j,1}^{'},\dots,A_{i,j,n/ab}^{'}$ be a partition of the sorted multiset $\{w(P_{i,j,\ell}) | \ell \in [n]\}$ into $n/ab$ multisets with a bound $x = 2n$ on its elements, proven to exist in Claim \ref{combi:partitioning}.
		
		Let $A_{i,j,1},\dots,A_{i,j,n/ab}$ be a partition of $[n]$ such that for every $k \in [n/ab]$, $A_{i,j,k}^{'}=\{w(P_{i,j,\ell}) | \ell \in A_{i,j,k}\}$.
		\label{procedure:ASBMM1.2:partitionOfPijkMatricies}.
		
		\caption{\textbf{Compute-Receiving}: Code for node $v \in \protect{[n]}$, which is also denoted $v_{i,j,k}$.}
		\label{procedure:ASBMM1.2}
	\end{algorithm}

	\begin{lemma}
		\label{lemma:ReceivingSolution}
		Algorithm \ref{procedure:ASBMM1.2} completes in $O(1)$ rounds, after which each node $v_{i,j,k}$ is assigned a subset $A_{i,j,k} \subseteq [n]$, \emph{s.t.} $\forall k \in [n/ab]$ it holds that $\sum_{\ell \in A_{i,j,k}}{w(P_{i,j,\ell})} \leq \frac{n}{ab}\sum_{\ell \in [n]}{w(P_{i,j,\ell})}+2n$.
	\end{lemma}
	
	\begin{proof}
		The loop of Line~\ref{line:BigLoop} provides each node of $N_{i',j'}$ with the number of non-zero elements in each column of $S_{i'}$ and each row of $T_{j'}$. This allows the nodes to compute the required communication costs in Line~\ref{line:wLoop}. Claim~\ref{combi:partitioning} implies that after executing Line~\ref{procedure:ASBMM1.2:partitionOfPijkMatricies}, each node $v_{i,j,k}$ is assigned a subset $A_{i,j,k} \subseteq [n]$, such that for every $k \in [n/ab]$ it holds that $\sum_{\ell \in A_{i,j,k}}{w(P_{i,j,\ell})} \leq \frac{n}{ab}\sum_{\ell \in [n]}{w(P_{i,j,\ell})}+2n$.
		
		Every node sends every other node exactly $4$ messages throughout the loop in Line~\ref{line:BigLoop}, while the remaining lines are executed locally for each node, without communication. As such, this completes in $O(1)$ rounds in total.
	\end{proof}

	\textbf{Resolve-Routing:} Roughly speaking, we solve this challenge by having the recipient of each possibly non-zero entry deduce which node is the sender of this entry, and inform the sender that it is its recipient, as follows. At the end of the execution of Compute-Sending in Algorithm~\ref{procedure:ASBMM1.1}, every node $v$ has at most two subsequences in $B_S(v)$ and at most two subsequences in $B_T(v)$. Moreover, the subsequence assignment is known to all nodes due to performing the same local computation in Line~\ref{procedure:ASBMM1.1BlockAssignmentStep}. On the other hand, upon completion of Algorithm~\ref{procedure:ASBMM1.2}, node $v_{i,j,k}$ is assigned the task of computing $P_{i,j,\ell}$ for every $\ell \in A_{i,j,k}$. For this, it suffices for $v_{i,j,k}$ to know the non-zero entries of column $\ell$ of $S_i$ and of row $\ell$ of $T_j$.
	
	Hence, in Resolve-Routing, given in Algorithm~\ref{procedure:ASBMM1.3}, node $v_{i,j,k}$ sends every index $\ell \in A_{i,j,k}$ to the nodes that hold subsequences of column $\ell$ in $S$ and row $\ell$ in $T$. Notice that $v_{i,j,k}$ does not know which indices inside these columns and rows are non-zero. However, the nodes which hold these subsequences have this information, and respond with the non-zero entries of the respective columns and rows that are part of $S_i$ or $T_j$. 

	\begin{algorithm}
		\ForEach{$\ell \in A_{i,j,k}$}
		{\label{procedure:ASBMM1.3Label1}
			\ForEach{node $u$ for which there exists $B \in B_S(u)$ such that $\ell(B)=\ell$}
			{
				\Send $\ell$ to node $u$\label{stepAbove1}
			}			
			\ForEach{node $u$ for which there exists $B \in B_T(u)$ such that $\ell(B)=\ell$}
			{
				\Send $\ell$ to node $u$\label{stepAbove2}
			}		
		}
		\label{procedure:ASBMM1.3Label2}
		\ForEach{message $\ell$ received from node $v_{i',j',k'}$ in Line~\ref{stepAbove1}}
		{
			\label{procedure:ASBMM1.3Label3}
			\ForEach{$B \in B_S(v)$}
			{	
				\Send $S[\SIndecies{i'}][\ell]\cap B$  to node $v_{i',j',k'}$\label{procedure:ASBMM1.3Label4}
			}
		}
		
		\ForEach{message $\ell$ received from node $v_{i',j',k'}$ in Line~\ref{stepAbove2}}
		{
			\label{procedure:ASBMM1.3Label5}
			\ForEach{$B \in B_T(v)$}
			{					
				\Send  $T[\ell][\TIndecies{j'}]\cap B$  to node $v_{i',j',k'}$\label{procedure:ASBMM1.3Label6}
			}
		}
		
		\caption{\textbf{Resolve-Routing}: Code for node $v \in \protect{[n]}$, which is also denoted $v_{i,j,k}$.}
		\label{procedure:ASBMM1.3}
	\end{algorithm}

	\begin{lemma}
		\label{lemma:KnowledgeSolution}
		Algorithm~\ref{procedure:ASBMM1.3} completes in $O(nz(S)\cdot b/n^2+nz(T)\cdot a/n^2 + 1)$ rounds, after which each node $v_{i,j,k}$ has $S[\SIndecies{i}][\ell]$ and $T[\ell][\TIndecies{j}]$, for every $\ell \in A_{i,j,k}$.
	\end{lemma}
	
	\begin{proofof}{Lemma~\ref{lemma:KnowledgeSolution}}
	By Lemma~\ref{lemma:SendingSolution}, every node knows $\ell(B), v(B)$ for every subsequence $B$, implying that Lines \ref{procedure:ASBMM1.3Label1}-\ref{procedure:ASBMM1.3Label2} can be executed. In Lines \ref{procedure:ASBMM1.3Label3}-\ref{procedure:ASBMM1.3Label6}, each node $v_{i,j,k}$ receives $S[\SIndecies{i}][\ell]$ and $T[\ell][\TIndecies{j}]$, for every $\ell \in A_{i,j,k}$, completing the correctness proof.
	
	For the round complexity, notice that by Lemma~\ref{lemma:SendingSolution}, for each node $u$, $B_S(u)$ contains entries from at most two distinct columns of $S$ and $B_T(u)$ contains entries from at most two distinct rows of $T$. Therefore, every node sends at most $4$ messages to every other node throughout Lines \ref{procedure:ASBMM1.3Label1} - \ref{procedure:ASBMM1.3Label2}. Thus, this part completes in $4$ rounds.
	
	We now show that Lines \ref{procedure:ASBMM1.3Label3} - \ref{procedure:ASBMM1.3Label6} complete in $O(nz(S)\cdot b/n^2+nz(T)\cdot a/n^2 + 1)$ rounds.
	Each node $v$ sends the entries of every $B \in B_S(v)$ to a single node in each of $b$ sets $N_{i',j'}$. Since $nz(B) \leq nz(S)/n+1$ by Claim~\ref{combi:SimpleByAvg}, this implies that each node sends at most $O(nz(S)\cdot b/n)$ entries in Lines \ref{procedure:ASBMM1.3Label3} - \ref{procedure:ASBMM1.3Label4}. A similar argument shows that each node sends at most  $O(nz(T)\cdot a/n)$ entries for each of the two $B \in B_T(v)$. In total, this sums up to sending at most $O(nz(S)\cdot b/n+nz(T)\cdot a/n)$ entries by each node. Likewise, we show that this is the number of entries that need to be received by each node. This is because the number of non-zero entries from $S$ and $T$ required for $v_{i,j,k}$ to compute the entries of the matrices $P_{i,j,\ell}$ for $\ell \in A_{i,j,k}$ is at most $(nz(S)/a + n)/(n/ab) + 2n +(ns(T)/b + n)/(n/ab) + 2n$, by Lemma~\ref{lemma:ReceivingSolution}. Due to the fact that $n/ab \geq 1$, the previous expression is bounded above by $(nz(S)/a)/(n/ab) + (nz(T)/b)/(n/ab) + 6n = nz(S)\cdot b/n+nz(T)\cdot a/n + 6n$.
	
	Finally, by Lenzen's routing scheme, Lines \ref{procedure:ASBMM1.3Label3} - \ref{procedure:ASBMM1.3Label6} complete in $O(nz(S)\cdot b/n^2+nz(T)\cdot a/n^2  + 1)$, completing the proof.
\end{proofof}

We can now wrap-up the proof of Theorem~\ref{theorem:MultiplyGivenNiceST}.

	\begin{proofof}{Theorem~\ref{theorem:MultiplyGivenNiceST}}
		Lemma~\ref{lemma:KnowledgeSolution} implies that each node $v_{i,j,k}$ has the required entries of $S$ and $T$ in order to compute $P_{i,j,\ell}$ for every $\ell \in A_{i,j,k}$. Lemma~\ref{lemma:afterExchange} then gives that Algorithm \algnice correctly produces a row of $P=ST$ for each node.
		
		Lemmas~\ref{lemma:SendingSolution} and \ref{lemma:ReceivingSolution} show that Compute-Sending and Compute-Receiving complete in $O(1)$ rounds.
		Lemma~\ref{lemma:KnowledgeSolution} gives the claimed round complexity of $O(nz(S)\cdot b/n^2+nz(T)\cdot a/n^2 + 1)$ for Resolve-Routing, giving the same total number of rounds for \algASBMMSubroutineOne.
		By Lemma~\ref{lemma:afterExchange}, the remainder of Algorithm \algnice completes in $O(n/ab + 1)$ rounds, completing the proof.
	\end{proofof}

\section{APSP and counting 4-cycles}
\label{sec:applications}
As applications of Algorithm \alg, we improve upon the state-of-the-art in the \clique model, for several fundamental graph problems, when considering sparse graphs. In Section~\ref{sec: SGD}, we utilize SMM alongside an additional algorithm for calculating the trace of the product of two matrices in order to count the number of 4-cycles of a given graph $G$. In Section~\ref{sec: APSP}, we utilize \alg for computing APSP in a way which is faster for some range of parameters that depends on the sparsity and the diameter of $G$.

In what follows, given a graph $G$, we denote by $m$ the number of its edges and by $A_G$ its adjacency matrix.

\subsection{Counting 4-cycles}
\label{sec: SGD}
In order to compute the number of $4$-cycles of $G$, it is sufficient to compute the trace\footnote{The \emph{trace} of a matrix is the sum of the entries that lie on its main diagonal.} of $A^4_G$ and the degrees of the nodes, as described in~\cite{Censor-HillelKK15}. This allows us to utilize Algorithm \alg for squaring the adjacency matrix $A_G$ and deducing the trace of $A^4_G$. It is noteworthy that we can avoid raising $A_G$ to the power of 4, and compute the trace of this power by only squaring $A_G$.

\begin{theorem-repeat}{theorem:counting}
	\TheoremCounting
\end{theorem-repeat}

\begin{proof}
	First, observe that given two $n\times n$ matrices $A,B$, it is possible to compute the trace of the matrix $A\cdot B$ in $O(1)$ communication rounds in the \clique model. This is done by redistributing the matrix $T$ across the nodes such that node $v$ holds column $v$ instead of row $v$ of $T$, and having each node $v \in [n]$ locally compute the diagonal entry $P[v][v]=S[v][*]\cdot T[*][v]$. Then, each node $v$ broadcasts $P[v][v]$ to all other nodes in a single round, and thus all nodes are able to sum these values and deduce $trace(AB)$.
	
	For counting $4$-cycles, the nodes first execute Algorithm \alg for obtaining $A_G^2$, and then compute $trace(A_G^2 \cdot A_G^2)$, as explained above. By a result of Alon et al.~\cite{AlonYZ97}, the number of $4$-cycles in $G$ equals $\frac{1}{8}(trace(A_G^4-\sum_{v \in [n]}{(2d_v^2-d_v)})$, where $d_v$ is the degree of $v$ in $G$. Since all nodes can obtain $d_v$ for all nodes $v$ in a single round of broadcasting the degrees, by Corollary~\ref{cor:m-non-zero}, this procedure completes in $O(m^{2/3}/n + 1)$ rounds.
\end{proof}

We remark that Alon et al.~\cite{AlonYZ97} has similar formulas for traces of larger values of $k$.

\subsection{APSP}
\label{sec: APSP}
In order to compute APSP for a given graph $G$, it is sufficient to compute $A^D_G$ over the min-plus semiring\footnote{In the min-plus semiring, $P[i][j]=\min_k{(S[i][k]+T[k][j])}$.}, where $D$ is the diameter of $G$. Notice that we cannot use the approach of~\cite{Censor-HillelKK15} which repeatedly squares the adjacency matrix, thus paying only a logarithmic overhead beyond a single multiplication, because powers of a sparse matrix may be dense. However, if we know $D$, then by repeatedly applying Corollary~\ref{cor:one-sparse} $D-1$ times, we can compute APSP in $O(D(m/n)^{1/3} + 1)$ rounds. In fact, any constant approximation of $D$ suffices, and hence we first run a simple BFS computation in order to obtain a $2$-approximation $\tilde{D}$ for $D$, which we then follow with raising $A_G$ to the power of $2\tilde{D}$. This gives the following.

\begin{theorem-repeat}{theorem:APSP}
	\TheoremAPSP
\end{theorem-repeat}

\section{Triangle Listing}
\label{sec:triangle-listing}

In \cite{Pandurangan0S16}, Pandurangan et al. show a randomized algorithm for listing all triangles, completing in $\tilde{O}(m/n^{5/3})$ rounds, w.h.p. Our contribution is a \emph{deterministic} algorithm which completes in $O(m/n^{5/3} + 1)$ rounds always. For dense graphs, this is not worse than the tight bound of $\Theta(n^{1/3})$, as given by the lower bounds of~\cite{IzumiG17, Pandurangan0S16}.

\begin{theorem-repeat}{theorem:triangleListing}
	\TheoremListing
\end{theorem-repeat}

In fact, our algorithm will apply also to directed graphs, being able to distinguish directed triangles. Each edge in the graph is oriented, and for every $v \in V$, let $d_{in}(v)$ and $d_{out}(v)$ denote the in and out degrees of node $v$, respectively. The undirected case follows easily, for example by imagining that every edge represents two edges in opposite directions.

Before giving the algorithm, we begin with some definitions and notations.
For two sets $S_1, S_2 \subseteq V$ we denote the set of edges from $S_1$ to $S_2$ by $E(S_1, S_2) = E\cap (S_1 \times S_2)$. When one of the sets is a singleton we abuse notation and write, e.g., $E(v, S_1)$. Our algorithm uses several partitions of sets of elements that are computed by the nodes after learning the amounts of elements, similarly to what we do in our matrix multiplication algorithm.

\begin{definition}
	\textbf{Equally Sized Partition}: A set of sets $\{V_1, \dots, V_t\}$ is a \emph{$t$-equally sized partition} of a set $V$ if $V = \bigcup_{i \in [t]} V_i$, $\forall i\neq j \in [t]: V_i \bigcap V_j = \emptyset$, and $ \forall i \in [t]: |V_i| = n/t$.
\end{definition}

For a partition $\{A_1, \dots, A_t\}$ of $V$, we sometimes need an \emph{internal numbering} of the nodes within each set $A_i$ which is uniquely determined by the node identifiers. Thus, every node in $V$ has a unique pair of indexes $(i,j)$ such that $a_{i,j}$ is the $j$-th node in the internal numbering of $A_i$.

We will use the following predefined notation: Let $\{D_1, \dots, D_{n^{2/3}}\}$ is a fixed globally known $n^{2/3}$-equally sized partition of $V$.
We denote $\alpha = m/n^{1/3} + n, \beta = m/n^{2/3} + n$.

Our algorithm will use, among other computations, three subroutines that we will describe separately. The first, $\Broadcast(m)$, consists of the node $v$ sending a message $m$ to every node in the graph. The second and third, $\LearnEdges(X)$ and $\LearnPaths(X)$ result in node $v$ gaining knowledge of all the elements in the set $X$. The latter are implemented in two different ways, for their two different uses. Notice that in \LearnPaths there is execution of the Resolve-Routing algorithm; in order to match the notation used in Algorithm \ref{procedure:triangleListing}, the indexing in lines 8, 11 of Resolve-Routing should be replaced with the edges exiting and entering node sets $V_{j_\delta}, V_{i_\delta}$, respectively.

\begin{algorithm}
	\Broadcast($d_{in}(v), d_{out}(v)$) \label{step1}
	
	All nodes partition $V$ into an $n^{1/3}$-equally sized partition $V_1, \dots, V_{n^{1/3}}$ such that for every $ i \in [n^{1/3}]$, $\sum_{u \in V_i} d_{in}(u)+ d_{out}(u) \leq 2\alpha$, proven to exist by Claim~\ref{combi:partitioning}. Denote $v=v_{i',j'}$. \label{step3}
	
	\ForEach{$j \in [n^{1/3}]$}{ \label{step4}
		\ForEach{$k \in [n^{2/3}]$}{ \label{step5}
			\Send $|E(v, V_j)|$ to $v_{i', k}$ \label{step6}
		}
	}
	
	\ForEach{$j \in [n^{1/3}]$}{ \label{step7}
		Compute $m_{i', j} = |E (V_{i'} , V_{j})| = \sum_{u \in V_{i'}}|E (u, V_j)|$ \label{step8}
		
		Nodes in $V_{i'}$ partition $V_{i'}$ into $N_{i', j, 1}, \dots, N_{i', j, \ceil{m_{i', j}/\beta}}$ such that for every $\ell \in [\ceil{m_{i', j}/\beta}]$, $|E (N_{i', j, \ell} , V_j)| \leq \beta$,  proven to exist by Claim~\ref{combi:partitioning}. \label{step9}
	}
	
	\ForEach{$i \in [n^{1/3}]$}{ \label{step10}
		\Send $\ceil{m_{i', 1}/\beta}, \dots, \ceil{m_{i', n^{1/3}}/\beta}$ to $v_{i, j'}$ \label{step11}
	}
	
	All nodes partition the set of all $N_{i, j, \ell}$ sets into two equally sized sets $A^1, A^2$ \label{step12}
	
	\ForEach{$t \in \{1, 2\}$}{ \label{step13}
		For every $k \in [n^{2/3}]$, uniquely assign at most one $N_{i, j, \ell} \in A^t$ to $D_k$, denoted by $N_{i_k, j_k, \ell_k}$ \label{step14}
		
		\LearnEdges $E (N_{i_\delta, j_\delta, \ell_\delta} , V_{j_\delta})$ \label{step15}
		
		\ForEach{$k \in [n^{2/3}]$}{ \label{step16}
			\ForEach{$u \in D_k$}{ \label{step17}
				\Send $|E  (V_{j_k} , v)| + |E  ( v , V_{i_k})|$ to $u$ \label{step18}
			}
		}
		
		Nodes in $D_\delta$ partition $V$ into $P_1, \dots P_{n^{1/3}}$ such that for every $d \in [n^{1/3}]$, $|E  (V_{j_\delta}, P_d)| + |E  (P_d , V_{i_\delta} )| \leq  4\beta$, proven to exist by Claim~\ref{combi:partitioning}. \label{step19}
		
		\LearnPaths$E  (V_{j_\delta} , P_{\delta'}) \cup E (P_{\delta'} , V_{i_\delta} )$ \label{step20}
		
		$v$ outputs a list of any triangles it sees \label{step21}
	}
	
	\algrule[1pt]
	\setcounter{AlgoLine}{0}
	\nonl
	\LearnEdges (Step \ref{step15}):
	
	Arbitrarily order the outgoing edges of $v$: $e_1 = (v, u_1), \dots, e_{d_{out}(v)} = (v, u_{d_{out}(v)})$ \label{step2.1}
	
	For every $ i \in [d_{out}(v)]$, create an \emph{information packet} $p_i = \{e_i, i', j, \ell, k\}$, where $j$ is such that $u_i \in V_j$, $\ell$ is such that $v \in N_{i', j, \ell}$, and $k$ is such that $N_{i', j, \ell}$ is assigned to $D_k$. \label{step2.2}
	
	Uniquely \emph{allocate} all packets such that each node gets at most $m/n + 1$ packets, proven to exist by Claim~\ref{combi:SimpleByAvg}. \label{step2.3}
	
	\ForEach{$i \in [d_{out}(v)]$}{ \label{step2.4}
		\Send $p_i$ to the node to which $p_i$ is allocated \label{stepSendPackets} \label{step2.5}
	}
	
	\ForEach{packet $p = \{e, i, j, \ell, k\}$ received in Step \ref{stepSendPackets}}{ \label{step2.6}
		
		\ForEach{$u \in D_k$}{ \label{step2.7}
			\Send $e$ to $u$ \label{step2.8}
		}
		
	}
	
	\algrule[1pt]
	\setcounter{AlgoLine}{0}
	\nonl
	\LearnPaths (Step \ref{step20}):
	
	Execute Compute-Sending using $S = T = A$, $a = b = n^{1/3}$
	
	Execute Resolve-Routing using $S = T = A$, $a = b = n^{1/3}, A_{i,j,k} = P_{\delta'}$

	\caption{Code for node $v \in \{1,\dots,n\}$, which is also denoted $d_{\delta, \delta'}$.}
	\label{procedure:triangleListing}
\end{algorithm}

\textbf{Algorithm Overview. } Let $E^3$ be the set of all ordered triplets of edges in $G$,  and let $Tri \subseteq E^3$ be the set of all directed triangles in $G$. Our goal is to distribute (perhaps multiple copies of) parts of $E$ across the nodes of the graph such that each $(e_1, e_2, e_3) \in Tri$ is known to at least one node. For intuition, observe the following simple algorithm for triangle listing, which also appears in~\cite{DolevLP12}: (1) arbitrarily partition the nodes of the graph into an $n^{1/3}$-equally sized partition $\{V_1, \dots, V_{n^{1/3}}\}$, (2) assign every ordered triplet of sets $(V_i, V_j, V_k)$ from the partition to a different node and have each node learn $E(V_i, V_j) \cup E(V_j, V_k) \cup E(V_k, V_i)$. Clearly, any triangle will be found by some node. 
This simple approach is costly due to the \emph{load imbalance} which it suffers from - for different triplets $(V_i, V_j, V_k)$, the number of edges in $E(V_i, V_j) \cup E(V_j, V_k) \cup E(V_k, V_i)$ may vary drastically, with some triplets being very dense while others are sparse.

Our approach is to utilize load balancing routing strategies from Algorithm \alg, in order to create a \emph{layered} partition of $V$ that ensures that every node learns roughly the same amount of edges. In more detail, instead of having a single node in charge of every triplet $(V_i,V_j,V_k)$, we take another (fixed) partition $\{D_1, \dots, D_{n^{2/3}}\}$ of $V$ which is an $n^{2/3}$-equally sized partition, and assign each \emph{pair} $(V_i,V_j)$ to a set of $n^{1/3}$ nodes in some $D_k$.
Then, within each $D_k$, the nodes again partition $V$ into $\{P_1, \dots, P_{n^{1/3}}\}$ in order to take into account the distribution of edges between the graph and the pair $(V_i, V_j)$ that is assigned to $D_k$. Here too, we get that all triangles in the graph are guaranteed to be listed. For the cost of communication to be very efficient, we choose the two non-fixed partitions carefully, according to information that the nodes exchange regarding amounts of edges between different sets in the graph.

Thus, Algorithm \ref{procedure:triangleListing} works as follows. We begin by partitioning $V$ into an $n^{1/3}$-equally sized partition $\{V_1, \dots, V_{n^{1/3}}\}$ of $V$ that has the property that for every $ i \in [n^{1/3}]$, $\sum_{u \in V_i} d_{in}(u)+ d_{out}(u) \leq 2\alpha$. Then, we take every pair of sets in the partition, $(V_i, V_j)$, and split the set $E(V_i, V_j)$ into smaller sets by creating $N_{i,j,\ell}$ sets such that $|E (N_{i, j, \ell} , V_j)| \leq \beta$, and each such $N_{i, j, \ell}$ is assigned to a predefined fixed set $D_k$ of $n^{1/3}$ unique nodes. Next, the nodes within $D_k$ compute a partition $\{P_1, \dots, P_{n^{1/3}}\}$ of $V$ such that for every $d \in [n^{1/3}]$, $|E  (V_{j}, P_d)| + |E  (P_d , V_{i} )| \leq  4\beta$. Finally, every node within $D_k$ is assigned one $P_d$ and learns all the edges in $E (N_{i, j, \ell} , V_{j}) \cup E  (V_{j} , P_{d}) \cup E (P_{d} , V_{i} )$.

The pseudo-code of the algorithm is described in Algorithm \ref{procedure:triangleListing}.

\begin{proofof}{Theorem~\ref{theorem:triangleListing}.}

	We need to show that every $(e_1, e_2, e_3) \in Tri$ is known to at least one node. Let $(v_1, v_2, v_3, v_1)$ be a directed triangle, with the edges $e_1 = (v_1, v_2), e_2 = (v_2, v_3), e_3 = (v_3, v_1)$. Observing the partition defined in Step \ref{step3}, let $i, j$ be the unique integers such that $v_1 \in V_i, v_2 \in V_j$. Notice that due to the loop in Steps \ref{step7}-\ref{step9}, there also exists an integer $\ell$ such that $v_1 \in N_{i, j, \ell}$ and that $N_{i, j, \ell} \subseteq V_i$. Further notice that in Step \ref{step14}, there is some unique integer $k$ such that the set $N_{i,j,\ell}$ is assigned to $D_k$. Next, in Step \ref{step19}, the nodes in $D_k$ agree on some partition of $V$ into $\{P_1, \dots, P_{n^{1/3}}\}$. Let $d$ be the unique integer such that $v_3 \in P_d$ for the specific $\{P_1, \dots, P_{n^{1/3}}\}$ partition agreed to by $D_k$. Observe the node $d_{k, d}$: in Step \ref{step15}, all the nodes in $D_k$, including $d_{k,d}$, learn all the edges $E(N_{i,j,\ell}, V_j)$; in Step \ref{step20}, $d_{k,d}$ learn all the edges $E(V_j, P_d) \cup E(P_d, V_i)$. Therefore, node $d_{k,d}$ is aware of the edges $e_1, e_2, e_3$ and so in Step \ref{step21} it outputs this triangle.
	
	For the number of rounds, it follows by Lenzen's routing scheme that all the steps of Algorithm \ref{procedure:triangleListing}, with the exception of the \LearnEdges and \LearnPaths instructions in Steps \ref{step15} and \ref{step20}, can be executed in $O(1)$ rounds in the \clique model. For \LearnEdges and \LearnPaths, a code inspection gives that no node sends or receives more than $O(\beta)$ messages, guaranteeing a total round complexity of $O(m/n^{5/3} + 1)$.
\end{proofof}

	\section{Discussion}
\label{sec:discussion}
	This work significantly improves upon the round complexity of multiplying two matrices in the distributed \clique model, for input matrices which are sparse. As mentioned, we are unaware of a similar algorithmic technique being utilized in the literature of parallel computing, which suggests that our approach may be of interest in a more general setting. The central ensuing open question left for future reserach is whether the round complexity of sparse matrix multiplication in the \clique can be further improved.
	
	Finally, an intriguing question is the complexity of various problems in the more general $k$-machine model ~\cite{KlauckNP015,Pandurangan0S16}, where the size of the computation clique is $k<<n$. The way of partitioning the data to the nodes is of importance. One may assume that the input to each node consists of $n/k$ unique consecutive rows of $S$ and $T$, and its output should be the corresponding $n/k$ rows of the product $P = S \cdot T$. Applying our algorithm in this setting gives a round complexity of $O(\min_{n\text{-split pairs } (a,b)} n^2/k^2 + nz(S)\cdot b/k^2+nz(T) \cdot a/k^2+n^2/kab + 1)$ rounds, which is $O(n^{2/3}\cdot nz(S)^{1/3} nz(T)^{1/3}/k^{5/3} + 1)$ rounds with the assignment $a = n^{2/3} k^{1/3} \cdot nz(S)^{1/3} / nz(T)^{2/3}$ and $b = n^{2/3} k^{1/3} \cdot nz(T)^{1/3} / nz(S)^{2/3}$. To see why, consider each node as simulating the behavior of $n/k$ \emph{virtual} nodes of the \clique model that belong to the same $N_{i,j}$ set. The round complexity of all steps of the algorithm grows by a multiplicative factor of $n^2/k^2$, apart from the steps in Algorithm~\ref{procedure:ASBMM} which grow only by a multiplicative factor of $n/k$, since part of the simulated communication consists of messages sent between virtual nodes that are simulated by the same actual node, and as such do not require actual communication. We ask whether this complexity can be improved for $k<<n$.
		
	\textbf{Acknowledgements:} The authors thank Seri Khoury, Christoph Lenzen, and Merav Parter for many useful discussions and suggestions. This project has received funding from the European Union's Horizon 2020 Research And Innovation Programe under grant agreement no.~755839. Supported in part by ISF grant 1696/14.
\bibliography{SparseMM}

\begin{thebibliography}{10}

\bibitem{AhmedHHRP17}
Md.~Salman Ahmed, Jennifer Houser, Mohammad~A. Hoque, Rezaul Raju, and Phil
  Pfeiffer.
\newblock Reducing inter-process communication overhead in parallel sparse
  matrix-matrix multiplication.
\newblock {\em {IJGHPC}}, 9(3):46--59, 2017.
\newblock URL: \url{https://doi.org/10.4018/IJGHPC.2017070104}, \href
  {http://dx.doi.org/10.4018/IJGHPC.2017070104}
  {\path{doi:10.4018/IJGHPC.2017070104}}.

\bibitem{AlonYZ97}
Noga Alon, Raphael Yuster, and Uri Zwick.
\newblock Finding and counting given length cycles.
\newblock {\em Algorithmica}, 17(3):209--223, 1997.
\newblock URL: \url{https://doi.org/10.1007/BF02523189}, \href
  {http://dx.doi.org/10.1007/BF02523189} {\path{doi:10.1007/BF02523189}}.

\bibitem{AmossenP09}
Rasmus~Resen Amossen and Rasmus Pagh.
\newblock Faster join-projects and sparse matrix multiplications.
\newblock In {\em Proceedings of the 12th International Conference on Database
  Theory (ICDT)}, pages 121--126, 2009.
\newblock URL: \url{http://doi.acm.org/10.1145/1514894.1514909}, \href
  {http://dx.doi.org/10.1145/1514894.1514909}
  {\path{doi:10.1145/1514894.1514909}}.

\bibitem{AzadBBDGSTW16}
Ariful Azad, Grey Ballard, Aydin Bulu{\c{c}}, James Demmel, Laura Grigori, Oded
  Schwartz, Sivan Toledo, and Samuel Williams.
\newblock Exploiting multiple levels of parallelism in sparse matrix-matrix
  multiplication.
\newblock {\em {SIAM} J. Scientific Computing}, 38(6), 2016.
\newblock URL: \url{https://doi.org/10.1137/15M104253X}, \href
  {http://dx.doi.org/10.1137/15M104253X} {\path{doi:10.1137/15M104253X}}.

\bibitem{BallardBDGLST13}
Grey Ballard, Aydin Bulu{\c{c}}, James Demmel, Laura Grigori, Benjamin
  Lipshitz, Oded Schwartz, and Sivan Toledo.
\newblock Communication optimal parallel multiplication of sparse random
  matrices.
\newblock In {\em Proceedings of the 25th {ACM} Symposium on Parallelism in
  Algorithms and Architectures, (SPAA)}, pages 222--231, 2013.
\newblock URL: \url{http://doi.acm.org/10.1145/2486159.2486196}, \href
  {http://dx.doi.org/10.1145/2486159.2486196}
  {\path{doi:10.1145/2486159.2486196}}.

\bibitem{BallardDKS16}
Grey Ballard, Alex Druinsky, Nicholas Knight, and Oded Schwartz.
\newblock Hypergraph partitioning for sparse matrix-matrix multiplication.
\newblock {\em {TOPC}}, 3(3):18:1--18:34, 2016.
\newblock URL: \url{http://doi.acm.org/10.1145/3015144}, \href
  {http://dx.doi.org/10.1145/3015144} {\path{doi:10.1145/3015144}}.

\bibitem{BulucG11}
Aydin Bulu{\c{c}} and John~R. Gilbert.
\newblock The combinatorial {BLAS:} design, implementation, and applications.
\newblock {\em {IJHPCA}}, 25(4):496--509, 2011.
\newblock URL: \url{https://doi.org/10.1177/1094342011403516}, \href
  {http://dx.doi.org/10.1177/1094342011403516}
  {\path{doi:10.1177/1094342011403516}}.

\bibitem{BulucG12}
Aydin Bulu{\c{c}} and John~R. Gilbert.
\newblock Parallel sparse matrix-matrix multiplication and indexing:
  Implementation and experiments.
\newblock {\em {SIAM} J. Scientific Computing}, 34(4), 2012.
\newblock URL: \url{https://doi.org/10.1137/110848244}, \href
  {http://dx.doi.org/10.1137/110848244} {\path{doi:10.1137/110848244}}.

\bibitem{Censor-HillelKK15}
Keren Censor{-}Hillel, Petteri Kaski, Janne~H. Korhonen, Christoph Lenzen, Ami
  Paz, and Jukka Suomela.
\newblock Algebraic methods in the congested clique.
\newblock In {\em Proceedings of the {ACM} Symposium on Principles of
  Distributed Computing (PODC)}, pages 143--152, 2015.
\newblock URL: \url{http://doi.acm.org/10.1145/2767386.2767414}, \href
  {http://dx.doi.org/10.1145/2767386.2767414}
  {\path{doi:10.1145/2767386.2767414}}.

\bibitem{CoppersmithW90}
Don Coppersmith and Shmuel Winograd.
\newblock Matrix multiplication via arithmetic progressions.
\newblock {\em J. Symb. Comput.}, 9(3):251--280, 1990.
\newblock URL: \url{https://doi.org/10.1016/S0747-7171(08)80013-2}, \href
  {http://dx.doi.org/10.1016/S0747-7171(08)80013-2}
  {\path{doi:10.1016/S0747-7171(08)80013-2}}.

\bibitem{DeveciTR18}
Mehmet Deveci, Christian Trott, and Sivasankaran Rajamanickam.
\newblock Multi-threaded sparse matrix-matrix multiplication for many-core and
  {GPU} architectures.
\newblock {\em CoRR}, abs/1801.03065, 2018.
\newblock URL: \url{http://arxiv.org/abs/1801.03065}, \href
  {http://arxiv.org/abs/1801.03065} {\path{arXiv:1801.03065}}.

\bibitem{DolevLP12}
Danny Dolev, Christoph Lenzen, and Shir Peled.
\newblock "tri, tri again": Finding triangles and small subgraphs in a
  distributed setting - (extended abstract).
\newblock In {\em Proceedings of the 26th International Symposium on
  Distributed Computing (DISC)}, pages 195--209, 2012.
\newblock URL: \url{https://doi.org/10.1007/978-3-642-33651-5_14}, \href
  {http://dx.doi.org/10.1007/978-3-642-33651-5_14}
  {\path{doi:10.1007/978-3-642-33651-5_14}}.

\bibitem{DruckerKO13}
Andrew Drucker, Fabian Kuhn, and Rotem Oshman.
\newblock On the power of the congested clique model.
\newblock In {\em Proceedings of the {ACM} Symposium on Principles of
  Distributed Computing (PODC)}, pages 367--376, 2014.
\newblock URL: \url{http://doi.acm.org/10.1145/2611462.2611493}, \href
  {http://dx.doi.org/10.1145/2611462.2611493}
  {\path{doi:10.1145/2611462.2611493}}.

\bibitem{LeGall12}
Fran{\c{c}}ois~Le Gall.
\newblock Faster algorithms for rectangular matrix multiplication.
\newblock In {\em Proceedings of the 53rd Annual {IEEE} Symposium on
  Foundations of Computer Science (FOCS)}, pages 514--523, 2012.
\newblock URL: \url{https://doi.org/10.1109/FOCS.2012.80}, \href
  {http://dx.doi.org/10.1109/FOCS.2012.80} {\path{doi:10.1109/FOCS.2012.80}}.

\bibitem{LeGall14a}
Fran{\c{c}}ois~Le Gall.
\newblock Powers of tensors and fast matrix multiplication.
\newblock In {\em International Symposium on Symbolic and Algebraic Computation
  (ISSAC)}, pages 296--303, 2014.
\newblock URL: \url{http://doi.acm.org/10.1145/2608628.2608664}, \href
  {http://dx.doi.org/10.1145/2608628.2608664}
  {\path{doi:10.1145/2608628.2608664}}.

\bibitem{LeGall16}
Fran{\c{c}}ois~Le Gall.
\newblock Further algebraic algorithms in the congested clique model and
  applications to graph-theoretic problems.
\newblock In {\em Proceedings of the 30th International Symposium on
  Distributed Computing (DISC)}, pages 57--70, 2016.
\newblock URL: \url{https://doi.org/10.1007/978-3-662-53426-7_5}, \href
  {http://dx.doi.org/10.1007/978-3-662-53426-7_5}
  {\path{doi:10.1007/978-3-662-53426-7_5}}.

\bibitem{LeGallU18}
Francois~Le Gall and Florent Urrutia.
\newblock Improved rectangular matrix multiplication using powers of the
  coppersmith-winograd tensor.
\newblock In {\em Proceedings of the Twenty-Ninth Annual {ACM-SIAM} Symposium
  on Discrete Algorithms (SODA)}, pages 1029--1046, 2018.
\newblock URL: \url{https://doi.org/10.1137/1.9781611975031.67}, \href
  {http://dx.doi.org/10.1137/1.9781611975031.67}
  {\path{doi:10.1137/1.9781611975031.67}}.

\bibitem{Gustavson78}
Fred~G. Gustavson.
\newblock Two fast algorithms for sparse matrices: Multiplication and permuted
  transposition.
\newblock {\em {ACM} Trans. Math. Softw.}, 4(3):250--269, 1978.
\newblock URL: \url{http://doi.acm.org/10.1145/355791.355796}, \href
  {http://dx.doi.org/10.1145/355791.355796} {\path{doi:10.1145/355791.355796}}.

\bibitem{IzumiG17}
Taisuke Izumi and Fran{\c{c}}ois~Le Gall.
\newblock Triangle finding and listing in {CONGEST} networks.
\newblock In {\em Proceedings of the {ACM} Symposium on Principles of
  Distributed Computing, {PODC} 2017, Washington, DC, USA, July 25-27, 2017},
  pages 381--389, 2017.
\newblock URL: \url{http://doi.acm.org/10.1145/3087801.3087811}, \href
  {http://dx.doi.org/10.1145/3087801.3087811}
  {\path{doi:10.1145/3087801.3087811}}.

\bibitem{KaplanSV06}
Haim Kaplan, Micha Sharir, and Elad Verbin.
\newblock Colored intersection searching via sparse rectangular matrix
  multiplication.
\newblock In {\em Proceedings of the 22nd {ACM} Symposium on Computational
  Geometry (SocG)}, pages 52--60, 2006.
\newblock URL: \url{http://doi.acm.org/10.1145/1137856.1137866}, \href
  {http://dx.doi.org/10.1145/1137856.1137866}
  {\path{doi:10.1145/1137856.1137866}}.

\bibitem{KlauckNP015}
Hartmut Klauck, Danupon Nanongkai, Gopal Pandurangan, and Peter Robinson.
\newblock Distributed computation of large-scale graph problems.
\newblock In {\em Proceedings of the Twenty-Sixth Annual {ACM-SIAM} Symposium
  on Discrete Algorithms, {SODA} 2015, San Diego, CA, USA, January 4-6, 2015},
  pages 391--410, 2015.
\newblock URL: \url{https://doi.org/10.1137/1.9781611973730.28}, \href
  {http://dx.doi.org/10.1137/1.9781611973730.28}
  {\path{doi:10.1137/1.9781611973730.28}}.

\bibitem{KoanantakoolABM16}
Penporn Koanantakool, Ariful Azad, Aydin Bulu{\c{c}}, Dmitriy Morozov,
  Sang{-}Yun Oh, Leonid Oliker, and Katherine~A. Yelick.
\newblock Communication-avoiding parallel sparse-dense matrix-matrix
  multiplication.
\newblock In {\em Proceedings of the {IEEE} International Parallel and
  Distributed Processing Symposium (IPDPS)}, pages 842--853, 2016.
\newblock URL: \url{https://doi.org/10.1109/IPDPS.2016.117}, \href
  {http://dx.doi.org/10.1109/IPDPS.2016.117}
  {\path{doi:10.1109/IPDPS.2016.117}}.

\bibitem{LazzaroVHS17}
Alfio Lazzaro, Joost VandeVondele, J{\"{u}}rg Hutter, and Ole Sch{\"{u}}tt.
\newblock Increasing the efficiency of sparse matrix-matrix multiplication with
  a 2.5d algorithm and one-sided {MPI}.
\newblock {\em CoRR}, abs/1705.10218, 2017.
\newblock URL: \url{http://arxiv.org/abs/1705.10218}, \href
  {http://arxiv.org/abs/1705.10218} {\path{arXiv:1705.10218}}.

\bibitem{Lenzen13}
Christoph Lenzen.
\newblock Optimal deterministic routing and sorting on the congested clique.
\newblock In {\em Proceedings of the {ACM} Symposium on Principles of
  Distributed Computing (PODC)}, pages 42--50, 2013.
\newblock URL: \url{http://doi.acm.org/10.1145/2484239.2501983}, \href
  {http://dx.doi.org/10.1145/2484239.2501983}
  {\path{doi:10.1145/2484239.2501983}}.

\bibitem{Nanongkai14}
Danupon Nanongkai.
\newblock Distributed approximation algorithms for weighted shortest paths.
\newblock In {\em Proceedings of the 46th ACM Symposium on Theory of Computing
  (STOC)}, pages 565--573, 2014.
\newblock URL: \url{http://doi.acm.org/10.1145/2591796.2591850}, \href
  {http://dx.doi.org/10.1145/2591796.2591850}
  {\path{doi:10.1145/2591796.2591850}}.

\bibitem{Pandurangan0S16}
Gopal Pandurangan, Peter Robinson, and Michele Scquizzato.
\newblock Tight bounds for distributed graph computations.
\newblock {\em CoRR}, abs/1602.08481, 2016.
\newblock URL: \url{http://arxiv.org/abs/1602.08481}, \href
  {http://arxiv.org/abs/1602.08481} {\path{arXiv:1602.08481}}.

\bibitem{SolomonikBVH17}
Edgar Solomonik, Maciej Besta, Flavio Vella, and Torsten Hoefler.
\newblock Scaling betweenness centrality using communication-efficient sparse
  matrix multiplication.
\newblock In {\em Proceedings of the International Conference for High
  Performance Computing, Networking, Storage and Analysis, {SC} 2017, Denver,
  CO, USA, November 12 - 17, 2017}, pages 47:1--47:14, 2017.
\newblock URL: \url{http://doi.acm.org/10.1145/3126908.3126971}, \href
  {http://dx.doi.org/10.1145/3126908.3126971}
  {\path{doi:10.1145/3126908.3126971}}.

\bibitem{SolomonikD11}
Edgar Solomonik and James Demmel.
\newblock Communication-optimal parallel 2.5d matrix multiplication and {LU}
  factorization algorithms.
\newblock In {\em Euro-Par 2011 Parallel Processing - 17th International
  Conference, Euro-Par 2011, Bordeaux, France, August 29 - September 2, 2011,
  Proceedings, Part {II}}, pages 90--109, 2011.
\newblock URL: \url{https://doi.org/10.1007/978-3-642-23397-5_10}, \href
  {http://dx.doi.org/10.1007/978-3-642-23397-5_10}
  {\path{doi:10.1007/978-3-642-23397-5_10}}.

\bibitem{Strassen69}
Volker Strassen.
\newblock Gaussian elimination is not optimal.
\newblock {\em Numerische Mathematik}, 13(4):354--356, 1969.
\newblock \href {http://dx.doi.org/10.1007/BF02165411}
  {\path{doi:10.1007/BF02165411}}.

\bibitem{Tiskin01}
Alexandre Tiskin.
\newblock All-pairs shortest paths computation in the {BSP} model.
\newblock In {\em Automata, Languages and Programming, 28th International
  Colloquium, {ICALP} 2001, Crete, Greece, July 8-12, 2001, Proceedings}, pages
  178--189, 2001.
\newblock URL: \url{https://doi.org/10.1007/3-540-48224-5_15}, \href
  {http://dx.doi.org/10.1007/3-540-48224-5_15}
  {\path{doi:10.1007/3-540-48224-5_15}}.

\bibitem{Williams12}
Virginia~Vassilevska Williams.
\newblock Multiplying matrices faster than coppersmith-winograd.
\newblock In {\em Proceedings of the 44th ACM Symposium on Theory of Computing
  (STOC)}, pages 887--898, 2012.
\newblock URL: \url{http://doi.acm.org/10.1145/2213977.2214056}, \href
  {http://dx.doi.org/10.1145/2213977.2214056}
  {\path{doi:10.1145/2213977.2214056}}.

\bibitem{YusterZ05}
Raphael Yuster and Uri Zwick.
\newblock Fast sparse matrix multiplication.
\newblock {\em {ACM} Trans. Algorithms}, 1(1):2--13, 2005.
\newblock URL: \url{http://doi.acm.org/10.1145/1077464.1077466}, \href
  {http://dx.doi.org/10.1145/1077464.1077466}
  {\path{doi:10.1145/1077464.1077466}}.

\end{thebibliography}


\end{document}